\documentclass{amsart} 
\usepackage{latexsym,amsmath,amssymb,color}
\usepackage{ifthen}
\usepackage{enumerate}
\usepackage{verbatim}
\usepackage{graphicx}

\newtheorem{thm}{Theorem}[section]
\newtheorem{lemma}[thm]{Lemma}

\newtheorem{defin}{Definition}[section]
\theoremstyle{definition}

\renewcommand{\Re}{\operatorname{\rm Re}\nolimits}
\renewcommand{\Im}{\operatorname{\rm Im}\nolimits}

\def \tr {\operatorname{tr}}

\def \restrict {\upharpoonright}
\def \Real {{\mathbb R}}
\def \Sphere {\mathbb{S}}
\def \Complex {\mathbb{C}}
\def \Natural {{\mathbb N}}
\def \mco  {\mathcal{O}}
\def \Sphere {{\mathbb S}}

\def \Integers {{\mathbb Z}}

\def \cd {d'}
\def \sign {\text{sign}}

\def \bbE {{\mathbb E}}

\newcommand{\R}{\mathbb{R}}

\newcommand{\Z}{\mathbb{Z}}

\newcommand{\N}{\mathbb{N}}

\newcommand{\C}{\mathbb{C}}

\newcommand{\beq}{\begin{equation}}
\newcommand{\eeq}{\end{equation}}
\newcommand{\ba}{\begin{array}}
\newcommand{\ea}{\end{array}}
\newcommand{\bea}{\begin{eqnarray}}
\newcommand{\eea}{\end{eqnarray}}

\title [Resonances in Even Dimensions]
{Maximal Order of Growth for the
Resonance Counting Functions for Generic Potentials in Even Dimensions}
   \author { T.\ J.\  Christiansen and P.\ D.\ Hislop}
\thanks{T.J.C.\ partially supported by NSF grant DMS 0500267, P.D.H.\
partially supported by NSF grant 0503784.}

\begin{document}


\begin{abstract}
We prove that the resonance counting functions for Schr\"odinger
operators $H_V = - \Delta + V$ on $L^2 ( \R^d)$, for $d \geq 2$ {\it
even}, with generic, compactly-supported, real- or complex-valued
potentials $V$, have the maximal order of growth $d$ on each sheet
$\Lambda_m$, $m \in \Z \backslash \{ 0 \}$,
of the logarithmic Riemann surface. We obtain this
result by constructing, for each $m \in \Z \backslash \{ 0 \}$,
a plurisubharmonic function from a
scattering determinant whose zeros on the physical sheet $\Lambda_0$
determine the poles on $\Lambda_m$.
We prove
that the order of
growth of the counting function is related to a suitable
estimate on this function that we establish for generic potentials.
We also show that for a potential that
is the characteristic function of a ball,
the resonance counting function is
bounded below by $C_m r^d$
 on each sheet $\Lambda_m$, $m \in \Z \backslash \{0\}$.
\end{abstract}

\maketitle
\section{Introduction}\label{intro1}

We study the distribution of scattering poles or resonances of
Schr\"odinger operators $H_V = - \Delta + V$ on $L^2 ( \R^d)$, for
$d \geq 2$ and {\it even}, and real- or complex-valued,
compactly-supported potentials $V \in L_0^\infty ( \R^d; \C)$. Let
$\chi_V \in C_0^\infty (\R^d)$ be a smooth, compactly-supported
function satisfying $\chi_V V = V$, and denote the resolvent of
$H_V$ by $R_V (\lambda) = (H_V - \lambda^2)^{-1}$ for $0<\arg \lambda
<\pi$. In the
even-dimensional case, the operator-valued function $\chi_V R_V
(\lambda) \chi_V$ has a meromorphic continuation to $\Lambda$,
 the
infinitely-sheeted Riemann surface of the logarithm. We
denote by $\Lambda_m$ the $m^{th}$ open sheet consisting of $z \in
\Lambda$ with $m \pi < \arg z < (m + 1) \pi$. We are interested in the
number of poles $n_{V,m} (r)$,
counted with multiplicity,
 of the continuation of the truncated resolvent $\chi_V R_V (\lambda) \chi_V$
on $\Lambda_m$ of modulus at most $r > 0$.
The {\it order of growth} of the resonance counting function $n_{V,m}(r)$ for $H_V$ on the $m^{th}$-sheet
is defined by
\beq\label{resonance1}
\rho_{V,m} \equiv \limsup_{r \rightarrow \infty} \frac{ \log n_{V,m}(r)}{ \log r}.
\eeq
It is known that $\rho_{V,m} \leq d$ for $d\geq2$ even \cite{[Vodev1],[Vodev2]}. We prove that generically
(in the sense of Baire typical)
the resonance counting function has the maximal order of growth $d$ on each
non-physical sheet.

\begin{thm}\label{main1}
Let $d \geq 2$ be even, and let $K \subset \R^d$ be a fixed, compact set with nonempty interior.
There is a dense $G_\delta$ set $\mathcal{V}_F (K) \subset L_0^\infty (K; F)$, for $F=\R$ or $F=\C$, such
that if $V \in \mathcal{V}_F (K)$, then $\rho_{V,m} = d$ for all $m \in \Z \backslash \{ 0 \}$.
\end{thm}

This result is the even-dimensional analog of our previous result \cite{ch-hi1} in the odd-dimensional
case. Roughly speaking, the theorem
states that most Schr\"odinger operators $H_V$
have the right number of resonances on each non-physical sheet.
As in \cite{{polar}}, the proof depends upon the
construction of a plurisubharmonic function from which we can recover
the order of growth
of the resonance counting function
on $\Lambda_m$.  This is more difficult than in the odd-dimensional case,
where one can use a Weierstrass factorization on the plane to help understand
the relation between the order of growth of the determinant of the scattering
matrix and the order of growth of its zero-counting function.  Moreover,
the even dimensional case requires the study of the poles of the resolvent far from the physical
sheet. For this, one cannot simply use the determinant of the scattering matrix on the
physical sheet.

In order to implement the argument
in \cite{polar}, we need
upper bounds
on $n_{V,m}(r)$ for all such $V$. In even dimensions, the
following upper bounds for any
$V \in L_0^\infty ( \R^d; F)$ were proven by Vodev \cite{[Vodev1],[Vodev2]}.
Let $\lambda_j$ be the poles of the continuation of
$\chi_V R_V (\lambda) \chi_V$, listed with multiplicity.
Vodev considered the following counting function:
\beq\label{vodev1}
N_V (r,a ) \equiv \{ \lambda_j \in \Lambda :
0 < |\lambda_j| \leq r, | \arg \lambda_j| \leq a \}.
\eeq
We use the notation $\langle w \rangle \equiv
( 1 + | w|^2)^{1/2}$. In our setting, Vodev proves that
\beq\label{vodev2}
N_V (r,a) \leq C a ( \langle r \rangle^d + ( \log a)^d),
~\mbox{for all} ~ r,a > 1.
\eeq
This implies that $n_{V,m} (r) \leq C_m \langle r \rangle^d$ on each sheet $m \neq 0$.

 We also
need an example of a potential in our class
for which the order of growth of the
resonance counting function
is $d$ on each sheet $\Lambda_m$, $m \in \Z^*=\Z\setminus\{0\}$.
\begin{thm}\label{main111}
Let $V (x) = V_0\chi_{[0,1]} (|x|)$, with $V_0>0$, be a real,
spherically symmetric, step potential.
For each $m \in \Z \backslash \{ 0 \}$, there is a constant
$C_m >0$, so that for
$r$ sufficiently large,
\beq\label{lowerbd000}
n_{V,m}(r)\geq C_m r^d + \mathcal{O} (r^{d-1}).
\eeq
\end{thm}
Asymptotics of the resonance counting function for Schr\"odinger
operators in {\em odd} dimension
with certain radial potentials were proved in \cite{[Zworski2]}, see also
\cite{stef06}.

The existence of resonances for $H_V$ in even dimensions was first proved
by Tang and S\'a Barreto \cite{[TSaB]}. They proved for $d \geq 4$
that there must be at least one resonance if
$V \in C_0^\infty ( \R^d; \R)$ and $V$ is nontrivial.
This was later strengthened by S\'a Barreto \cite{[SaB1]}, who proved,
 for the same class
of potentials,
 a global lower bound on the number of resonances.
S\'a Barreto defined $N_V (r)$
to be the number of poles $\lambda$ of the continued, truncated resolvent, including multiplicities,
satisfying $1/r < \lambda < r$, with
$| \arg \lambda | < \log r$.
He proved that $N_V(r)$
grows more quickly than $\log r / (\log \log r)^p$, for any $p > 1$,
proving that there are infinitely-many resonances.
His work does not establish the existence
of infinitely-many resonances on each non-physical sheet, a result that follows
from Theorem \ref{main1},
for ``generic'' potentials.  In contrast, there are examples of
complex-valued potentials $V$ in dimension $d\geq 2$ for which there are
no resonances of $H_V$ away from the origin \cite{sownr}.

Earlier literature on the distribution of resonances for
Schr\"odinger operators in
the even-dimensional case includes Intissar's paper
\cite{[Intissar1]}. Intissar defined, for any $\epsilon > 0$ and $r > 1$,
a resonance counting function $N(\epsilon, r) \equiv \{ \lambda_j ~|~ r^{-\epsilon} <
| \lambda_j | < r^\epsilon , ~| \arg \lambda_j| < \epsilon \log r \}$.
For even dimensions $d \geq 4$, and for any $\epsilon \in (0 , \sqrt{2} /2 )$, he proved the
non-optimal
polynomial upper bound
$N(\epsilon, r) \leq C_\epsilon \langle r \rangle^{d+1}$.
Zworski \cite{[Zworski3]}
proved a Poisson formula for resonances in even dimensions.
Shenk and Thoe \cite{[ST1]}
proved the meromorphic continuation of the $S$-matrix for
potential scattering onto $\Lambda$.
In a related work, Tang \cite{[Tang1]} proved the existence
of infinitely-many resonances
for metric scattering
in even dimensions.

For surveys of related results in odd dimensions, we refer the reader to
\cite{[Zworski4]} and \cite{[Vodev3]}, in addition to the references in \cite{ch-hi1}.
We refer the reader to \cite{tjc3} for generic-type results for resonances in other situations.
The notion of Baire typical or generic potentials
was used by B.\ Simon in the study of singular continuous spectra for Schr\"odinger operators
in \cite{simon1}.


\subsection{Contents}
In section \ref{preliminaries1}, we first reduce the problem
of counting resonances on the $m^{th}$-sheet to one of counting the zeros of
a holomorphic function $f_m (\lambda)$ in the upper-half complex plane. We then
review aspects of the theory of plurisubharmonic functions.
Section \ref{complexanalysis} is devoted to the construction
of a plurisubharmonic function
associated with $f_m (\lambda)$
and families of potentials $V(z;x)$,
parameterized by $z \in \Omega \subset \C^{d'}$.
The main technical result on the order of growth of the
counting function for the zeros of $f_m (z, \lambda)$ is Theorem \ref{main-psh1}.
The main theorem on the distribution of resonances is proved in section \ref{mainth1}.
The proof of Theorem \ref{main111} on
the lower bound on the number of resonances on each sheet $\Lambda_m$, $m \in \Z \backslash \{ 0 \}$,
for the step potential with the correct exponent $d$, is
given in section \ref{examples}.
The appendix in section \ref{app-bounds} contains
the details of the uniform asymptotic expansions of Bessel and
 Hankel functions,
due to Olver \cite{[Olver1],[Olver2],[Olver3]}, and their application
to the location of the zeros.

\noindent
{\bf Acknowledgments.}  The first author is grateful for the
hospitality of MSRI during the completion of this paper, and to MSRI
and the MU Research Leave for partial support during this time.

\section{Reduction to an Operator on $\Lambda_0$ and other Preliminaries}\label{preliminaries1}

For $0<\arg \lambda <\pi$,
we denote by $R_0(\lambda)=(-\Delta -\lambda^2)^{-1}$, where $\Delta $ is
the nonpositive Laplacian on $\Real^d$.  It is well-known that for $d
\geq 4$ even and any compactly-supported, smooth function $\chi \in C_0^{\infty}(\Real^d)$,
the cut-off resolvent $\chi R_0(\lambda)\chi $ has an analytic continuation to $\Lambda$,
the logarithmic cover of the plane (see, for example, \cite{{lrb},[ST1]}). For $d=2$
there is a logarithmic singularity at $\lambda =0$. This is easily seen from the
representation
\beq\label{resolv1}
\chi R_0(\lambda) \chi = E_1(\lambda) + ( \lambda^{d-2} \log \lambda) \; E_2(\lambda),
\eeq
where $E_1(\lambda)$ and $E_2(\lambda)$ are entire operator-valued functions
and for $d=2$, the operator $E_2(\lambda=0)$ is a finite-rank operator.
This representation follows from the formula for the Green's function
\beq\label{resolv2}
R_0 (\lambda) = \frac{i}{4} \left( \frac{ \lambda}{2 \pi |x-y|} \right)^{(d-2)/ 2}
H_{(d-2)/2}^{(1)} ( \lambda |x-y|),
\eeq
where the Hankel function of the first kind is defined by $H_\nu^{(1)} (z) = J_\nu (z)
+ i N_\nu (z)$, and the expansion of the Neumann Bessel function $N_\nu (z)$, when $\nu \in \N$.
We shall continue to
denote this continuation by $\chi R_0(\lambda)\chi $.  We shall use the following
key identity, for $m\in \Integers$, that follows from (\ref{resolv2}) and (\ref{hankel11}),
\begin{equation}\label{eq:difference}
R_0(e^{im\pi}\lambda)=R_0(\lambda)-m T(\lambda)
\end{equation}
where $T(\lambda)$ has Schwartz kernel
\begin{equation}\label{Top1}
T(\lambda, x,y)= i \pi (2 \pi)^{-d} \lambda^{d-2}
\int_{\Sphere^{d-1} }e^{i\lambda (x-y)\cdot \omega}d\omega
\end{equation}
\cite[Section 1.6]{lrb}.
We note that for any $\chi \in C_0^{\infty}(\Real^d)$, $\chi T(\lambda)
\chi$ is a holomorphic trace-class valued operator for $\lambda \in \C$.

Let $V\in L^{\infty}_0 (\Real^d)$, where we suppress the notation $F$
when there is no need to distinguish real- or complex-valued potentials.
When $0<\arg \lambda <\pi$,
we denote by $R_V(\lambda)=(-\Delta +V-\lambda^2)^{-1}$.
Let $\chi_V \in C_0^\infty ( \R^d)$ be a compactly supported function
with $\chi_V V = V$. The resonances are the poles of the
 meromorphic continuation
of $\chi_V R_V(\lambda) \chi_V$ to $\Lambda$. The set of
resonances is independent of the
cut-off function in the class described here.
By the second resolvent formula, we have
\beq\label{resolvent1}
\chi_V R_V ( \lambda ) \chi_V ( 1 + V R_0 ( \lambda ) \chi_V)=
\chi_V R_0 (\lambda) \chi_V,
\eeq
and the poles of $\chi_V R_V (\lambda) \chi_V$ correspond,
with multiplicity, to the
zeros of $I+VR_0(\lambda)\chi_V$.

We can reduce the analysis of the zeros
of the continuation of $I+VR_0(\lambda)\chi_V$ to $\Lambda_m$
to the analysis of zeros of a related operator on $\Lambda_0$.  This
is very similar to the technique used by Froese in odd dimensions
in \cite{froese}.
Using (\ref{eq:difference}), if $0<\arg \lambda<\pi$ and $m\in \Integers$, then
$e^{im\pi} \lambda \in \Lambda_m$, and
\begin{align*}
I+VR_0(e^{im \pi}\lambda)\chi & = I+V(R_0(\lambda)-m T(\lambda))\chi_V \\
& = (I+VR_0(\lambda)\chi_V)(I-m (I+VR_0(\lambda)\chi_V)^{-1}V T(\lambda)\chi_V).
\end{align*}
For any fixed $V\in L^{\infty}_0 (\Real^d)$,
there are only finitely many poles of $(I+VR_0(\lambda) \chi_V)^{-1}$
with $0<\arg \lambda<\pi$.  Thus
\begin{equation}\label{eq:fmfirst}
f_m(\lambda)=\det(I -m (I+VR_0(\lambda)\chi_V)^{-1}V T(\lambda)\chi_V)
\end{equation} is
a holomorphic function of $\lambda$
when $0<\arg \lambda<\pi$ and
$|\lambda|>c_0 \langle \|V\|_{L^{\infty}}\rangle$.
Moreover, with at most a finite number of exceptions, the zeros
of $f_m(\lambda)$ with $0<\arg \lambda<\pi$ correspond, with multiplicity,
to the poles of $\chi_V R_V(\lambda) \chi_V$ with $m \pi <\arg \lambda < (m+1)\pi$.
Henceforth, we will consider the function $f_m (\lambda)$,
for $m \in \Z \backslash \{ 0 \}$, in a neighborhood
of  $\Lambda_0$.
We write $\Z^*$ for $\Z \backslash \{ 0 \}$.

We mention certain symmetries of the resonances, identified as poles of the meromorphic
continuation of the $S$-matrix, on various sheets that
were discussed in \cite{[ST1]}.
For $d$ even and $V$ real, if $\lambda \in \Lambda_m$ is a scattering
pole, then $\lambda_S = - \overline{\lambda} = \overline{ e^{-i \pi} \lambda} = |\lambda|
e^{i ( \pi - \arg \lambda)} \in \Lambda_{-m}$
is also a scattering pole. This follows from the identity for the $S$-matrix:
$S ( \overline{\lambda})^*
= 2 I - S(e^{i \pi} \lambda)$.
This result can also be seen from (\ref{eq:fmfirst}). If $\lambda \in \Lambda_m$
is a resonance, then $\tilde{\lambda} = e^{-i m \pi} \lambda \in \Lambda_0$ satisfies
$f_m(\tilde{\lambda})
=0$. For the symmetric point $\lambda_S = - \overline{\lambda} \in \Lambda_{-m}$, we
set $\tilde{\lambda_S} = e^{i m \pi} \lambda_S \in \Lambda_0$.
Using the identity $T (- \overline{\lambda}) = - \overline{T(\lambda)}$ for $0 < \arg \lambda < \pi$,
it easily follows that for $V$ real, $f_{-m}(\tilde{\lambda_S}) = 0$.
Note that in odd dimensions and $V$ real, the scattering poles in the lower-half complex
plane are symmetric about the imaginary axis as follows from the identity $S(\overline{\lambda})^*
= S(-\lambda)$.

We now turn to some preliminary analysis
of the cut-off resolvent and the operator $T(\lambda)$ defined in (\ref{Top1}).
We will always work with a fixed, but arbitrary, compact subset $K \subset \R^d$,
with nonempty interior,
and potentials $V \in L_0^\infty ( K ; F)$, for $F=\R$ or $F=\C$.
Our cut-off functions $\chi \in C_0^\infty ( \R^d)$ satisfy $\chi = 1$ on $K$
so $\chi V = V$ for all $V \in L_0^\infty (K; F)$.
Recall the notation $\langle x \rangle = ( 1 + |x|^2)^{1/2}$.
We begin with $L^2$-estimates.

\begin{lemma}\label{l:basics}
Suppose $\chi\in C_0^{\infty}(\Real^d)$ and $|\Im \lambda |\leq \alpha$, for any $\alpha > 0$.
Then,
\beq\label{eq:t-est1}
\|\chi T(\lambda)\chi\|_{L^2\rightarrow L^2}
\leq C\langle \lambda \rangle^{-1},
\eeq
where the constant $C$ depends on $\alpha$ and $\chi$.
Moreover, if $V\in L^{\infty}_{0}(K; F)$,
then there
is a constant $c_0$ such that
\begin{equation}\label{eq:invbound}
\|(I+VR_0(\lambda)\chi)^{-1}\|_{L^2 \rightarrow L^2}\leq 2
\end{equation}
when $-\pi/2 <\arg \lambda < 3\pi/ 2$, $|\Im \lambda |
\leq \alpha$, and $|\lambda|\geq c_0 \langle \|V\|_{L^{\infty}}\rangle$.
Moreover, for $\lambda$ in this same set,
$$\|I- (I+VR_0(\lambda)\chi)^{-1}\|_{L^2\rightarrow L^2}
\leq C \langle \|V\|_{L^{\infty}}\rangle \langle \lambda \rangle ^{-1}.$$
In particular, the bound (\ref{eq:invbound}) holds on $\Lambda_0$
when $|\lambda|> c_0\langle \|V\|_{L^{\infty}}\rangle.$  The constants
$C$ and $c_0$ depend on $\alpha$ and $\chi$.
\end{lemma}
In particular, (\ref{eq:invbound}) means that for
$|\lambda|\geq c_0 \langle \|V\|_{L^{\infty}}\rangle$ and satisfying the
other conditions, $I+VR_0(\lambda)\chi$ is invertible on $L^2$.
\begin{proof}
The first claim is proved in \cite{[zworski]} and published in
\cite[section 2, page 255]{[burq1]}.
Using that $\|\chi R_0(\lambda)\chi\|_{L^2\rightarrow L^2}=
\mathcal{O}(\langle \lambda \rangle^{-1})$ when $0<\arg \lambda <\pi$, we have
from
(\ref{eq:difference}) and (\ref{eq:t-est1}) that
\begin{equation}\label{eq:R0bound}
\|\chi R_0(\lambda)\chi\|_{L^2\rightarrow L^2}
= \mathcal{O}(\langle \lambda \rangle^{-1})
\end{equation} for
$- \pi/2 <\arg \lambda < 3\pi / 2$, $|\Im \lambda |
\leq \alpha$.  Then, there is a $c_0>0$ so that when
$|\lambda|\geq c_0 \langle \|V\|_{L^{\infty}}\rangle$, and
$- \pi/2 <\arg \lambda < 3\pi /2$, $|\Im \lambda |\leq \alpha$,
$\|V R_0(\lambda)\chi\|_{L^2\rightarrow L^2} \leq 1/2.$
For such values of $\lambda$, $\|(I+VR_0(\lambda)\chi)^{-1}\|_{L^2 \rightarrow L^2}\leq 2$.
The last claim follows from the series representation of $(I+VR_0(\lambda)
\chi)^{-1}$, using (\ref{eq:R0bound}).
\end{proof}

We now turn to trace norm estimates.
We denote by $\| \cdot \|_1$ the trace norm and by $\| \cdot  \|_2$ the Hilbert-Schmidt norm.

\begin{lemma}\label{l:fmbasic}
Suppose $V\in L^{\infty}_0 (K; F)$.
Then for $-\pi/4\leq \arg \lambda \leq 5\pi/4$, $|\Im \lambda|\leq
2$, $|\lambda|>c_0\langle \|V\|_{L^{\infty}}\rangle$, there is a $C>0$
so that
$$\| VT(\lambda)\chi\|_1\leq C\langle \|V\|_{L^{\infty}}\rangle
\langle \lambda \rangle^{d-2}$$
and
$$\| \frac{d}{d\lambda}VT(\lambda)\chi\|_1\leq C\langle \|V\|_{L^{\infty}}\rangle
\langle \lambda \rangle^{d-2}.$$
Moreover, there is a $c_{0,m}$ so that if
$|\lambda|>c_{0,m}\langle \|V\|_{L^{\infty}}\rangle$,
and $\lambda$ is in the same set as above,
\beq\label{Fbdd1}
\|(I-m (I+VR_0(\lambda)\chi)^{-1}V T(\lambda)\chi)^{-1}\|
_{L^2\rightarrow L^2}
\leq 2.
\eeq
\end{lemma}
\begin{proof}
We recall that
$\|AB\|_1 \leq \|A\|_2 \|B\|_2$,
and that the cut-off operator $T(\lambda)$ has a factorization
as
\begin{align*}
V T(\lambda ) \chi  = V\chi T(\lambda)\chi  = i \pi (2 \pi)^{-d} \lambda^{d-2}
V (\bbE^{\chi}(\lambda))^t \bbE^{\chi}(-\lambda)\\
\bbE^{\chi}(\lambda):L^2(\Real^d)\rightarrow L^2(\Sphere^{d-1}),\;
\bbE^{\chi}(\lambda,\omega, x)=\exp(-i\lambda x\cdot \omega )\chi(x).
\end{align*}
Since
$$\|\bbE^{\chi}(\lambda)\|^2_2
= \int\int  |\bbE^{\chi}(\lambda,\omega, x)|^2 d\omega dx
\leq C e^{C|Im \lambda|}$$
we obtain that
$$\| V T(\lambda ) \chi\|_1 \leq C \langle \|V\|_{L^{\infty}}\|\rangle
~|\lambda|^{d-2}e^{C|Im \lambda|}$$
with perhaps a different constant $C$.  A similar proof gives the estimate
for the derivative.
Using Lemma \ref{l:basics}, we can find a $c_{0,m}$ so that for
$|\lambda|>c_{0,m}\langle \|V\|_{L^{\infty}}\rangle$,
$$\|m (I+VR_0(\lambda)\chi)^{-1}V T(\lambda)\chi\|_{L^2\rightarrow L^2}
\leq 1/2.$$
Then
$$\|(I-m (I+VR_0(\lambda)\chi)^{-1}V T(\lambda)\chi)^{-1}\|
_{L^2\rightarrow L^2}
\leq 2.$$
\end{proof}

We also need the following result on the function $f_m(\lambda)$
defined in (\ref{eq:fmfirst}).

\begin{lemma}\label{l:fmsimilar}
Suppose $V\in L^{\infty}_0 (K;F)$.
Then there is a constant $c_m>0$ so
that if $\lambda$ lies in the set
$$\{ \lambda\in \Lambda : |\Im  \lambda| \leq 2 ,\;
\; | \arg \lambda | <\pi/4,
\; |\lambda|\geq c_m\langle \|V\|_{L^{\infty}}\rangle\}$$
then
$$ |f_m(\lambda e^{i\pi})/f_m(\lambda)|
\leq \exp(C \langle \|V\|_{L^{\infty}} \rangle \langle \lambda \rangle^{d-2})$$
and
$$|f_m(\lambda )/f_m(\lambda e^{i\pi})| \leq
\exp(C \langle \|V\|_{L^{\infty}} \rangle \langle \lambda \rangle^{d-2})$$
for some $C>0$.  The constants
$c_m$ and  $C$ are independent of $V$, but they depend on
the set $K$, $\chi$, and $m$.
\end{lemma}

\begin{proof}
We shall write
$$F_V(\lambda)=(I+VR_0(\lambda)\chi)^{-1}$$
to make the equations less cumbersome.
Note that
\begin{align*}
f_m(e^{i\pi} \lambda)& =
\det(I -m  (I+VR_0(e^{i\pi}\lambda)\chi)^{-1}V T(\lambda e^{i\pi})\chi)\\
& = \det (I-mF_V(e^{i\pi}\lambda)VT(e^{i\pi}\lambda)\chi)\\
&
= \det (I-mF_V(\lambda)VT(\lambda)\chi) \\ & \hspace{3mm}\times
\det \left( I+m(I-mF_V(\lambda)VT(\lambda)\chi)^{-1}
\left(F_V(\lambda)VT(\lambda)
-F_V(\lambda e^{i\pi})VT(\lambda e^{i\pi})\right)\chi\right)
\\ & =  f_m(\lambda)\det \left( I+m(I-mF_V(\lambda)VT(\lambda)\chi)^{-1}\left(
F_V(\lambda)VT(\lambda)
-F_V(\lambda e^{i\pi})VT(\lambda e^{i\pi})
\right)\chi\right).
\end{align*}
First, by Lemma \ref{l:basics} (\ref{eq:invbound}), we have for $j=0$ or $j=1$,
$\| F_V(e^{ij \pi} \lambda)\|_{L^2\rightarrow L^2} \leq 2$.
Second, by Lemma \ref{l:fmbasic} (\ref{Fbdd1}), there is a $c_{0,m}>0$ so that
when $|\lambda|\geq c_{0,m} \langle \|V\|_{L^{\infty}} \rangle$,
$\|(I-mF_V(\lambda)VT(\lambda)\chi)^{-1}\|_{L^2\rightarrow L^2}\leq 2.$
Third, by Lemma \ref{l:fmbasic}, for $j=0$ or $j=1$,
$\|VT(\lambda e^{ij\pi })\chi\|_1\leq C \|V\|_{L^{\infty}}\langle
\lambda\rangle^{d-2}$.  Finally, using $|\det(1+A)|\leq \exp(\|A\|_1)$,
we find that $|f_m(e^{i\pi} \lambda)/f_m(\lambda)| \leq
\exp( C\langle\|V\|_{L^{\infty}}\rangle \langle \lambda \rangle ^{d-2})$
for some constant $C>0$.
A similar technique gives the result for
$|f_m(\lambda)/f_m(e^{i\pi}\lambda)|$.
\end{proof}

The following lemma will be used in proving that a function we construct
is plurisubharmonic.
\begin{lemma}\label{l:fmint}
If $| \Im u| \leq 1$, $|  \arg u | \leq \pi/4$,
$|u|\geq c_m\langle \|V\|_{L^{\infty}}\rangle$, then
$$\left|\int_0^{\pi}\log|f_m(ue^{i\theta})|d\theta
- \int_0^{\pi}\log|f_m(|u|e^{i\theta})|d\theta\right|
\leq C \langle\|V\|_{L^{\infty}}\rangle ^2\langle u\rangle^{d-3}.
$$
\end{lemma}

\begin{proof}  We write $u=|u|e^{i\phi}$.  Then
\begin{align*} &
\int_0^{\pi}\log|f_m(ue^{i\theta})|d\theta \\ &
= \int_0^{\pi}\log|f_m(|u|e^{i\theta+i\phi})|d\theta \\
& = \int_0^{\pi}\log|f_m(|u|e^{i\theta})|d\theta
+\int_{\pi}^{\pi+\phi}\log|f_m(|u|e^{i\theta})|d\theta
- \int _{0}^{\phi}\log|f_m(|u|e^{i\theta})|d\theta.
\end{align*}
But
$$
\int_{\pi}^{\pi+\phi}\log|f_m(|u|e^{i\theta})|d\theta
- \int _{0}^{\phi}\log|f_m(|u|e^{i\theta})|d\theta
=  \int_0^{\phi}\log \left|
\frac {f_m(|u|e^{i(\theta+\pi)})}{f_m(|u|e^{i\theta})}
\right| d\theta$$
Using Lemma \ref{l:fmsimilar}, when $|u|$ is sufficiently large,
\begin{align*}
 \sign (\phi) \int_0^{\phi}\log \left|
\frac {f_m(|u|e^{i(\theta+\pi)})}{f_m(|u|e^{i\theta})}\right|
d\theta  & \leq C|\phi|\langle \|V\|_{L^{\infty}}\rangle (1+|u|^{d-2})
\\ &
\leq C |u|^{-1}\langle \|V\|_{L^{\infty}}\rangle (1+|u|^{d-2})
\end{align*}
where we use the fact that $|\phi|\leq C |u|^{-1}$ when $|u|$ is large.
We may also write
\begin{align*} &
\sign (\phi)\left( \int_{\pi}^{\pi+\phi}\log|f_m(|u|e^{i\theta})|d\theta
- \int _{0}^{\phi}\log|f_m(|u|e^{i\theta})|d\theta \right)
\\ &
= -\sign (\phi) \int_0^{\phi}\log \left|
\frac {f_m(|u|e^{i\theta})}{f_m(|u|e^{i(\theta+\pi)})}\right|
d\theta  \\ &
\geq -C|\phi|\langle \|V\|_{L^{\infty}}\rangle  \langle u\rangle ^{d-2}
\\ &
\geq -C|u|^{-1} \langle |V\|_{L^{\infty}}\rangle \langle u \rangle ^{d-2}
\end{align*}
again using Lemma \ref{l:fmsimilar}.
This finishes the proof of the lemma.
\end{proof}


\section{Order of Growth and Plurisubharmonic Functions}\label{complexanalysis}

We establish the main technical results in this section. It is the analog of section 2 of \cite{polar}
for the even-dimensional case. In the first subsection, we review the notion of
plurisubharmonic functions, pluripolar sets, and order,
and prove
Lemma \ref{l:horder}.
In the second subsection,
we construct a plurisubharmonic function associated with $f_m ( \lambda)$ to which we
will apply this lemma in order to estimate the order of growth of the resonance counting function.
The main result is Theorem \ref{main-psh1}.


\subsection{Review of some complex analysis}

For the convenience of the reader, we recall some basic notions used in
\cite{polar} that can be found in the book \cite{l-g}.
A domain $\Omega \subset \C^m$ is an open connected set.

\begin{defin}
A real-valued function $\phi (z)$ taking values in $[ -\infty, \infty)$ is
{\em plurisubharmonic} in
a domain $\Omega \subset \C^m$, and we write $\phi \in PSH ( \Omega)$, if:
\begin{itemize}
\item $\phi$ is upper semicontinuous and $\phi \not\equiv - \infty$;
\item for every $z \in \Omega$, for every $w \in \C^m$,  and every $r > 0$ such that
$\{ z + u w : |u| \leq r, u \in \C \}
\subset \Omega$
we have
\beq\label{subMV1}
\phi (z) \leq \frac{1}{2 \pi} \int_0^{2 \pi} ~\phi ( z + r e^{i \theta}w ) ~d \theta .
\eeq
\end{itemize}
\end{defin}

A function $\phi$ on $\Omega \subset \C^m$ is {\it locally plurisubharmonic}
on $\Omega$ if $\phi$ is upper semicontinuous and $\phi \not\equiv - \infty$ and
if for each $z \in \Omega$ there is a radius $b(z) > 0$ such that for all $w \in \C^m$
so that $|w| < b(z)$, $z + w e^{i \theta} \in \Omega$ and
\beq\label{locsubMV1}
\phi (z) \leq \frac{1}{2 \pi} \int_0^{2 \pi} ~\phi ( z + w e^{i \theta} ) ~d \theta .
\eeq
We recall that every locally plurisubharmonic function on a domain $\Omega$ is in $PSH (\Omega)$
\cite[Proposition I.19]{l-g}.

\begin{defin}\label{pluripolar1}
A set $E \subset \C^m$ is {\em pluripolar} if for each $a \in E$ there is a neighborhood
$V_a$ containing $a$ and a function $\phi_a \in PSH(V_a)$ such that $E \cap V_a \subset \{
z \in V_a : \phi_a (z) = - \infty \}$.
\end{defin}

\begin{defin}\label{defn-order}
For $r>0$, let $s(r)>0$ be a monotone increasing function of $r$.  If
$$\lim \sup _{r\rightarrow \infty}\frac{\log s(r)}{\log r}= \mu <\infty,$$
then $s(r)$ is said to be of {\em order} $\mu$.
\end{defin}

For a function $h$ holomorphic in
\begin{equation}\label{eq:hholo}
\{ \lambda \in\Complex:\; |\lambda|\geq R,\; \Im \lambda \geq 0\}
\end{equation}
and $r>R$,
define $n_{+,R}(h,r)$ to be the number of zeros of $h$, counted with multiplicities, in the closed
upper half plane with norm between $R$ and $r$, inclusive.
\begin{lemma}\label{l:horder}
Let $R>0$, let $h$ be holomorphic in the set in (\ref{eq:hholo}),
and suppose, in addition, that $h$ has only finitely many
zeros on the real axis.
Suppose that for some $p>0$, and for some $\epsilon>0$, we have
$$
\int_R^{r}\left|\frac{h'(s)}{h(s)}\right| ds
= \mathcal{O}(r^{p-\epsilon})\; \text{and}\;
\int_{-r}^{-R}\left|\frac{h'(s)}{h(s)}\right| ds
= \mathcal{O}(r^{p-\epsilon}).
$$
Then $n_{+,R}(h,r)$ has order $p$ if and only if
$$\lim \sup _{r\rightarrow \infty}
\frac{\log \int_0^{\pi}\log|h(re^{i\theta})|d\theta}{\log r}=p.$$
\end{lemma}
\begin{proof}
Notice that the order of $n_{+,R_1}(h,r)$ is independent of the
choice of $R_1\geq R$, since $n_{+,R_1}(h,r)=
n_{+,R_2}(h,r)+ \mathcal{O}(1)$ if $R_1,\; R_2\geq R$.  Thus we may
assume, without loss of generality, that $h$ has no zeros
with norm $R$ and none on the semi-axes $s > R$ and $s < -R$.
Then, using  the principle of the argument and the
Cauchy-Riemann equations,  for $t>R$
\begin{align*}
n_{+,R}(h,t) = &\frac{1}{2\pi}
\Im \left(\int_{-t}^{-R}\frac{h'(s)}{h(s)}ds
+ \int_{R}^{t}\frac{h'(s)}{h(s)}ds\right)\\ & +
\frac{1}{2\pi}\int_0^{\pi}t\frac{d}{dt}\log |h(te^{i\theta})|d\theta
-\frac{1}{2\pi}\int_0^{\pi}R\frac{d}{dR}\log |h(Re^{i\theta})|d\theta.
\end{align*}
Now, just as in the proof of Jensen's
equality, we divide by $t$ and integrate from $R$ to $r$ to obtain
\begin{align}\label{eq:jensenanalog}
\int_R^r\frac{n_{+,R}(h,t)}{t}dt =&  \frac{1}{2\pi}
\Im \left(\int_R^r t^{-1}\int_{-t}^{-R}\frac{h'(s)}{h(s)}dsdt
+\int_R^r t^{-1} \int_{R}^{t}\frac{h'(s)}{h(s)}dsdt\right)\nonumber \\ &
+
\frac{1}{2\pi}\int_0^{\pi}\log |h(re^{i\theta})|d\theta
- \frac{1}{2\pi}\int_0^{\pi}\log |h(Re^{i\theta})|d\theta \nonumber\\ &
-\frac{1}{2\pi}R\log(r/R)
\int_0^{\pi}\frac{d}{dR}\log |h(Re^{i\theta})|d\theta.
\end{align}
By our assumptions,
$$\int_R^r t^{-1}\int_{-t}^{-R}\frac{h'(s)}{h(s)}dsdt
= \mathcal{O}(r^{p-\epsilon})
\; \text{and}\; \int_R^r t^{-1} \int_{R}^{t}\frac{h'(s)}{h(s)}dsdt
= \mathcal{O}(r^{p-\epsilon}).$$
Thus
$$\lim\sup_{r\rightarrow\infty} \frac{\log \int_R^r\frac{n_{+,R}(h,t)}{t}dt}
{\log r}=p$$
if and only if
$$\lim \sup _{r\rightarrow \infty}
\frac{\log \int_0^{\pi}\log|h(re^{i\theta})|d\theta}{\log r}=p.$$
But
$n_{+,R}(h,t)$ and $\int_R^r\frac{n_{+,R}(h,t)}{t}dt $ have the
same order, proving the lemma.
\end{proof}


\subsection{Construction of a Plurisubharmonic Function}\label{psh1}

Let $\Omega \subset \Complex^{\cd}$ be an open connected set, let
$z\in \Omega $ and $x\in \Real^d$. We let $\Z^* \equiv \Z \backslash \{0\}$.
Throughout this
section we assume that  $V(z)=V(z,x)$ is a function which
has the following properties:
\begin{itemize}
\item For $z\in \Omega$, $V(z,\cdot )\in L^{\infty}(\Real^d)$.
\item  The function $V(z,x)$ is holomorphic in $z\in \Omega$.
\item There is a compact set $K_1 \subset \R^d$ so that
for $z\in \Omega$, $V(z,x)=0$ if $x\in \Real^d\setminus K_1$.
\end{itemize}
We refer to these three conditions as {\it Assumptions (V)}. In
our application in section \ref{mainth1}, we take $d' = 1$.

For $z\in \Omega$ and fixed $\chi\in C_0^{\infty}(\Real^d)$, with
$\chi\equiv 1$ on $K_1$, we set, in analogy with (\ref{eq:fmfirst}),
\begin{equation}\label{eq:fmz}
f_m(z,\lambda)=\det(I -m (I+V(z) R_0(\lambda)\chi)^{-1} V(z) T(\lambda) \chi).
\end{equation}
For $m\in \Integers^*$,
and $0 < \epsilon < 1$, we define the function $g_{m, \epsilon}$
by
\begin{equation}\label{eq:gmedef}
g_{m,\epsilon}(z,u)=\int_0^{\pi}\log |f_m(z,ue^{i\theta})|d\theta
+ \log |e^{u^{d-\epsilon}}|
\end{equation}
where a branch of $u^{d-\epsilon}$ is chosen so that $u^{d-\epsilon}
\in \Real $ when $u\in \Real_+$, and $u^{d-\epsilon}$ is holomorphic
for $|\arg u |<\pi/2$. We show that $g_{m,\epsilon}$ is a
plurisubharmonic function on an appropriate domain in $\Omega \times
\C$. The additional term $\log |e^{u^{d-\epsilon}}|$ in
(\ref{eq:gmedef}) is useful in order to determine certain growth
properties as shown in Lemma \ref{l:wheremaxgm}. We will use the
notation $\Omega' \Subset \Omega$ to mean that the subset
$\Omega'\subset \Omega$ is open and connected, with
$\overline{\Omega'}\subset \Omega$ compact.

\begin{lemma}\label{l:gmpsh}
For any $\Omega' \Subset \Omega$ and the constant $c_0$ is as in Lemma \ref{l:basics} with $\alpha =5$,
the function $g_{m, \epsilon} (z,u)$ defined in (\ref{eq:gmedef}) is plurisubharmonic for
$(z,u)\in \Omega'
\times U_{\Omega'}$ where
$$U_{\Omega'}= \{ u\in \Lambda:
\; |\Im u | < 2, \; \Re u >c_0\max_{z\in\Omega'}
\langle \|V(z)\|_{L^{\infty}}\rangle,
 \;  | \arg u | <\pi/4\}.
$$
\end{lemma}
\begin{proof}
Since $u^{d-\epsilon}$ is holomorphic on
$U_{\Omega'}$ (which can be identified with an unbounded rectangle in the complex plane),
then $\log |e^{u^{d-\epsilon}}|$ is plurisubharmonic in
$\Omega'\times U_{\Omega'}$.  Moreover, by Lemma \ref{l:basics},
$f_m(z, u e^{i\theta})$, defined in (\ref{eq:fmz}), is holomorphic on $\Omega'\times U_{\Omega'}$
when $0\leq \theta \leq \pi$.
Thus $\log |f_m(z,u e^{i\theta})|$ is plurisubharmonic on
$\Omega'
\times U_{\Omega'}$.
It then follows by \cite[Proposition I.14]{l-g} that
$\int_0^{\pi}\log |f_m(z,u e^{i\theta})|d\theta$ is either plurisubharmonic
on
$\Omega'
\times U_{\Omega'}$ or identically $-\infty$ there.  It is easy to see
that it is not identically $-\infty$.  Since the sum of two
plurisubharmonic functions is plurisubharmonic, we prove the lemma.
\end{proof}

In order to make notation less complicated, we
define, for $\Omega'\Subset \Omega$, the constant
$$V_{M,\Omega'}=\max _{z\in \Omega'}\|V(z, \cdot)\|_{L^{\infty}}.$$

\begin{lemma}\label{l:wheremaxgm} Let $g_{m,\epsilon}$ be
as defined by (\ref{eq:gmedef}), and let
$\Omega' \Subset \Omega$ and $U_{\Omega'}$ be as in Lemma \ref{l:gmpsh}.
Then there is an $r_{m}>0$ so that for $z \in \Omega'$,
$$\left( \max_{\substack{ |\Im u| \leq 1, \; \Re u>0 \\
 |u|= r,\; u\in U_{\Omega'}}}
g_{m,\epsilon}(z,u) \right)
> \left( \max _{\substack{\Im u = \pm 1, \; \Re u>0 \\
 |u|= r,\; u\in U_{\Omega'}}}
g_{m,\epsilon}(z,u)\right)$$
when $r\geq \langle V_{M, \Omega'} \rangle^{1/(1-\epsilon)}r_m.$
The constant $r_m$ depends on $m$, $K_1$ and $\chi$, but not on $V$.
\end{lemma}
\begin{proof}
Let $u=re^{i\phi}$, where $|\phi | < \pi/4$ and $\Im (re^{i\phi})=\pm 1$.
We shall show, in fact, that
$g_{m,\epsilon}(z,r)
>g_{m,\epsilon}(z,re^{i\phi})$, indicating that the maximum along the arc of $|u| = r$
with $| \Im u  | \leq 1$ occurs along the positive real axis at $u = r$.
Note that
\begin{align*}
|e^{u^{d-\epsilon}}|\restrict_{u=re^{i\phi}}&=\exp(\Re(r^{d-\epsilon}
e^{i\phi(d-\epsilon)})\\ &
=\exp(r^{d-\epsilon}\cos((d-\epsilon)\phi)).
\end{align*}
Thus
\begin{align*}
\log|e^{u^{d-\epsilon}}|\restrict_{u=r}
-\log |e^{u^{d-\epsilon}}|\restrict_{u=re^{i\phi}} &
= r^{d-\epsilon}-r^{d-\epsilon}\left(1-\frac{(d-\epsilon)^2\phi^2}{2}
+ O(\phi^4)\right)\\
& = r^{d-\epsilon}\frac{(d-\epsilon)^2\phi^2
}{2}
+ O(r^{d-\epsilon}\phi^{4}).
\end{align*}

Combining this with the
definition of $g_{m,\epsilon}$ and results of Lemma \ref{l:fmint}, we see
that
\begin{align*}
|g_{m,\epsilon}(z,r)-g_{m,\epsilon}(z, re^{i\phi}) -
r^{d-\epsilon}\frac{(d-\epsilon)^2\phi^2}{2}|\leq
C \left( r^{d-\epsilon}\phi^{4}+ \langle V_{M, \Omega'} \rangle
\langle r \rangle ^{d-3}\right).
\end{align*}
Since $|\phi|\approx r^{-1}$ when $r$ is sufficiently large, we may,
by choosing
$r_m$ sufficiently large,
ensure that
$$\frac{r^{d-\epsilon} (d-\epsilon)^2\phi^2}{2}>
2(C(r^{d-\epsilon}\phi^{4}+ \langle V_{M, \Omega'}\rangle
\langle r \rangle ^{d-3})
$$ when
$r\geq \langle V_{M, \Omega'} \rangle^{1/(1-\epsilon)}r_m$.
When this holds,
$g_{m,\epsilon}(z,r)-g_{m,\epsilon}(z,re^{i\phi})>0$, proving the lemma.
\end{proof}


\begin{lemma}\label{l:maxrbig}
 For an open set $\Omega'\Subset \Omega$, there
is a $\tilde{r}_{m,\epsilon}(\Omega', V)$ such that if
$r\geq\tilde{r}_{m,\epsilon}(\Omega', V)$,
$$g_{m,\epsilon}(z,r)>\max_{
\substack{|\Im u|\leq 1,\; |\arg u|\leq \pi/4\\
 \Re u>0,\; |u|=r_m(\langle V_{M,\Omega'}\rangle^{1/(1-\epsilon)}+1)}}
g_{m,\epsilon}(z,u)$$
for $z\in \Omega'$, where $r_m$ is defined in Lemma \ref{l:wheremaxgm}.
\end{lemma}
\begin{proof}
Let $\Omega''$ be an open set such that
$\Omega'\Subset \Omega''\Subset \Omega$ and $\langle
V_{M,\Omega''}\rangle^{1/(1-\epsilon)}<
\langle V_{M,\Omega'}\rangle^{1/(1-\epsilon)}+1$.
Since $g_{m,\epsilon}(z,u)$ is plurisubharmonic on
$$\Omega''\times \{ u\in \Complex :\; |\Im u|\leq 1,\; |\arg u|\leq \pi/4,\;
 \Re u>0,\; |u|>r_m\langle V_{M,\Omega''}\rangle^{1/(1-\epsilon)}\},$$
we have
\begin{equation}\label{eq:gsmallmax}
\max_{\substack{z\in \Omega',\; |\Im u|\leq 1,\; |\arg u|\leq \pi/4\\
 \Re u>0,\; |u|>r_m(\langle V_{M,\Omega'}\rangle^{1/(1-\epsilon)}+1)}}g_{m,\epsilon}(z,u)
<\infty.
\end{equation}

Let $z_0\in\overline{\Omega'}$.  Then we can find
$R_{z_0}>\max(1, r_m(\langle V_{M,\Omega'}\rangle^{2/(2-\epsilon)}+1))$ such
that $f_{m}(z_0,R_{z_0}e^{i\theta})$ has no zeros for $0\leq \theta\leq \pi$.
Additionally, we can find $\Omega_{z_0}'$, $\Omega_{z_0}''$, both open,
such that
$\Omega_{z_0}'\Subset \Omega_{z_0}''\Subset \Omega''$ and
$f_m(z,R_{z_0}e^{i\theta})\not = 0$ for $z\in \Omega_{z_0}''$, $0\leq \theta \leq \pi$.  Then $\log |f_m(z,R_{z_0}e^{i\theta})|$ and its derivatives
are continuous for such $z$ and $\theta$, since $f_m$ is holomorphic and
nonzero.

Using (\ref{eq:jensenanalog}), with $f_m$ in place of $h$,  and $R_{z_0}$ in
place of $R$, for $r>R_{z_0}$
\begin{align}\label{eq:fmlb}
\int_0^{\pi}\log|f_m(z,r e^{i\theta})|d\theta
\geq &
-\Im \left(\int_{R_{z_0}}^{r} t^{-1}\int_{-t}^{-R_{z_0}}
\frac{\frac{d}{ds}f_m (z,s)}{f_m(z,s)}dsdt
+\int_{R_{z_0}}^{r} t^{-1} \int_{R_{z_0}}^{t}\frac{\frac{d}{ds}f_m(z,s)}
{f_m(z,s)}dsdt\right)\nonumber \\ &
+ \int_0^{\pi}\log |f_m(z,R_{z_0}e^{i\theta})|d\theta \nonumber\\ &
+R_{z_0}\log(r/R_{z_0})
\int_0^{\pi}\frac{d}{dR}\log |f_m(z,Re^{i\theta})|\restrict_{R=R_{z_0}}d\theta.
\end{align}
There is a $C_{z_0}$ such that for $z\in \Omega'_{z_0}$,
\begin{equation}\label{eq:ubint}
\int_0^{\pi}\log |f_m(z,R_{z_0}e^{i\theta})|d\theta
+ R_{z_0}
\int_0^{\pi}\frac{d}{dR}\log |f_m(z,Re^{i\theta})|\restrict_{R=R_{z_0}}d\theta
\leq C_{z_0}
\end{equation}
since we stay in a region with $f_m$ nonzero.

Next, use that
$$\frac{d}{d\lambda}\det(I+A(\lambda))
= \tr \left((I+A(\lambda))^{-1}\frac{d}{d\lambda}A(\lambda)\right)$$
applied to $f_m(z,s)$.
Thus
\begin{equation}\label{eq:fmtrace}
\frac{\frac{d}{ds}f_m(z,s)}
{f_m(z,s)}= -m\tr \left(
(I -m (I+VR_0(\lambda)\chi)^{-1} V(z) T(\lambda) \chi)^{-1}
 \frac{d}{ds}\left(
 (I+VR_0(s)\chi)^{-1}V T(s)\chi\right)\right).\end{equation}
When $\arg s =0$ or $\arg s =\pi$ and $\langle s\rangle > c_{0,m}\langle
\|V(z)\|_{L^{\infty}}\rangle$,
$$\|(I -m (I+VR_0(\lambda)\chi)^{-1} V(z) T(\lambda) \chi)^{-1}
\|_{L^2\rightarrow L^2}\leq C $$
and $$
 \left \|\frac{d}{ds}
 (I+VR_0(s)\chi)^{-1}V T(s)\chi\right\|_{1}
\leq C\langle \|V(z)\|_{L^{\infty}}\rangle
\langle s \rangle^{d-2}$$
by Lemma \ref{l:fmbasic}.
Using this and (\ref{eq:fmtrace}), and increasing $C_{z_0}$
if necessary, we get that
\begin{equation}\label{eq:bd1}\left|
\int_{R_{z_0}}^{r} t^{-1}\int_{-t}^{-R_{z_0}}
\frac{\frac{d}{ds}f_m (z,s)}{f_m(z,s)}dsdt\right|
\leq C_{z_0}\langle \|V(z)\|_{L^{\infty}}\rangle
\langle r \rangle^{d-1}
\end{equation}
and
\begin{equation}\label{eq:bd2}
\left| \int_{R_{z_0}}^r t^{-1} \int_{R_{z_0}}^{t}\frac{\frac {d}{ds}f_m(z,s)}
{f_m(z,s)}dsdt\right|\leq C_{z_0}\langle \|V(z)\|_{L^{\infty}}\rangle
\langle r \rangle^{d-1}
\end{equation} for $z\in \Omega'_{z_0}$.
Using (\ref{eq:ubint}), (\ref{eq:bd1}), and (\ref{eq:bd2})
 in (\ref{eq:fmlb}), with perhaps a new $C_{z_0}$, we find that
$$\int_0^{\pi}\log|f_m(z,r e^{i\theta})|d\theta
\geq -C_{z_0}  \langle\|V(z)\|_{L^{\infty}}\rangle
\langle r \rangle^{d-1} .$$
Thus, for $z\in \Omega'_{z_0}$,
$$g_m(z,r )\geq r^{d-\epsilon}- C_{z_0}  \langle\|V(z)\|_{L^{\infty}}\rangle
\langle r \rangle^{d-1}.$$

By construction, $\overline{\Omega'}\subset \cup_{z_0\in \overline{\Omega'}}
\Omega'_{z_0}$ so that $\{\Omega'_{z_0}\}_{z_0\in \overline{\Omega'}}$
forms a cover of $\overline{\Omega'}$.
  Since $\overline{\Omega'}$ is a compact set, there
is a finite subcover.  Denote the corresponding $z_0$'s for this finite subcover
$z_1,\; z_2, \ldots , z_N$.
If we set $C_{\Omega'}=\max C_{z_j}$, $j=1, \ldots, N$, we have
\begin{equation}\label{eq:glb}
g_m(z,r )\geq r^{d-\epsilon}- C_{\Omega'}  \langle\|V(z)\|_{L^{\infty}}\rangle
\langle r \rangle^{d-1}
\end{equation}
for $z\in\Omega'$.
But then there is a $\tilde{r}_m(\Omega', V)$ so that
the right hand side of (\ref{eq:glb}) is bigger than (\ref{eq:gsmallmax}) whenever
$r\geq\tilde{r}_m(\Omega', V)$.
This finishes the proof.
\end{proof}

\begin{lemma}\label{l:tildem}
Let $g_{m,\epsilon}(z,u)$ be as in (\ref{eq:gmedef}) and
$\Omega'\Subset \Omega$.  Set
$$\tilde{M}_{m,\epsilon,\Omega'}(z,w)=
\max_{\substack{|\Im u |\leq 1, \; |\arg u|\leq \pi/4\\
 r_m(\langle V_{M,\Omega'}\rangle^{1/(1-\epsilon)}+1)\leq
|u| \leq |w|}}
g_{m,\epsilon}(z,u).
$$
Here $r_m$ is as in Lemma \ref{l:wheremaxgm}.
Then
 $\tilde{M}_{m, \epsilon,\Omega'}$ is plurisubharmonic on $$\Omega' \times
\{ w\in \Complex: \;  \tilde{r}_{m,\epsilon}(\Omega',V)<|w|\}$$
with $\tilde{r}_{m,\epsilon}(\Omega',V)$ as in
Lemma \ref{l:maxrbig}.\end{lemma}
The proof of this lemma resembles the proof of
\cite[Lemma 2.1]{polar}.
\begin{proof}
It is clear that $\tilde{M}_{m,\epsilon,\Omega'}\not \equiv -\infty$.
Since $g_{m,\epsilon}$ is upper semicontinuous, so is $\tilde{M}_{m, \epsilon,
\Omega'}$.
Since $g_{m,\epsilon}(z,u)$ is plurisubharmonic on
\begin{equation}\label{eq:uset}
\{| \Im u |\leq 1, \; |\arg u|\leq \pi/4,\;
 r_m(\langle V_{M,\Omega'}\rangle^{1/(1-\epsilon)}+1)\leq|u|\leq |w|,\;
\}
\end{equation}
its maximum is attained on the boundary.
So suppose $|w|>\tilde{r}_{m,\epsilon}(\Omega',V)$ and
$\tilde{M}_{m,\epsilon,\Omega'}(z,w)= g_{m,\epsilon}(z,u_0)$, with
$u_0$ on the boundary of the set in (\ref{eq:uset}).  By Lemma
\ref{l:wheremaxgm} and the choice of the set over which we take the
maximum of $g_{m,\epsilon}$, $-1<\Im u_0<1$.  By Lemma \ref{l:maxrbig},
$g_{m,\epsilon}(z,|w|)>g_{m,\epsilon}(z,u)$
when $|u|=r_m(\langle V_{M,\Omega'}\rangle^{1/(1-\epsilon)}+1)$,
$|\arg u|\leq \pi/4$, $|\Im u|\leq 1$.  Hence $|u_0|=|w|.$

For $v\in \Complex^{\cd+1}$  write
$v=(v',v_{\cd+1})\in \Complex^{\cd}\times \Complex$.  There is
a $\rho(z,w)$, so that for all $\theta\in [0,2\pi]$
and all $v\in \Complex^{\cd+1}$ with $|v|<\rho(z,w)$,
$z+e^{i\theta}v'\in \Omega'$ and $u_0+v_{\cd+1}e^{i\theta}$ lies
in the set in (\ref{eq:uset}).  Then, using the fact that
$g_{m,\epsilon}$ is plurisubharmonic by Lemma \ref{l:gmpsh}, we have
\begin{align*}
\tilde{M}_{m,\epsilon,\Omega'}(z,w) &
= g_{m,\epsilon}(z,u_0)\\
& \leq (2\pi)^{-1}\int_0^{2\pi} g_{m,\epsilon}(z+e^{i\theta }v',u_0+
e^{i\theta}\frac{u_0}{w}v_{\cd+1}) d\theta\\
& \leq (2\pi)^{-1}\int_0^{2\pi} \tilde{M}_{m,\epsilon,\Omega'}
(z+e^{i\theta }v',w+
e^{i\theta}v_{\cd+1}) d\theta.
\end{align*}
This shows that $\tilde{M}_{m,\epsilon,\Omega'}$ is
locally plurisubharmonic at $(z,w)$.  Since we may do the
same thing for any $z\in \Omega'$, $w\in \Complex$ with $|w|>
\tilde{r}_{m,\epsilon}(\Omega',V)$, $\tilde{M}_{m,\epsilon,\Omega'}$
is  plurisubharmonic.
\end{proof}

\begin{lemma}\label{l:M_m}
Let $\Omega' \Subset \Omega$
and let  $\tilde{M}_{m,\epsilon,\Omega'}$
be as in Lemma \ref{l:tildem}.
The function $M_{m , \epsilon, \Omega'}(z,w)$ defined on $\Omega ' \times \C$ by
$$M_{m,\epsilon,\Omega'}(z,w)=
\left\{ \begin{array}{l l}
\max (1,\tilde{M}_{m,\epsilon, \Omega'}(
z,\tilde{r}_m(\Omega',V)
+1) , & \text{if}\;
|w|\leq \tilde{r}_m(\Omega',V)
+1\\
\max (1,\tilde{M}_{m,\epsilon,\Omega'}(z,w)),& \text{if}\; |w|\geq
\tilde{r}_m(\Omega',V) ,
\end{array}
\right.
$$
is plurisubharmonic on $\Omega'\times \Complex$.
\end{lemma}
\begin{proof}
This proof follows just as the proof of \cite[Lemma 2.2]{polar}, if
one uses in addition the fact that the supremum of two plurisubharmonic
functions is plurisubharmonic (\cite[Proposition I.3]{l-g}).
\end{proof}

Note that the
dependence of $M_{m, \epsilon, \Omega'}(z,w)$ on $w$ is only through
the norm $|w|$. It follows from Lemmas \ref{l:tildem} and \ref{l:M_m} that
the function $r\mapsto M_{m, \epsilon, \Omega'}(z,r)$ is monotone increasing. The following lemma
demonstrates the relationship between
the order of $n_{V(z)}(r)$ and the order of $r\mapsto M_{m, \epsilon, \Omega'}(z,r)$.

\begin{lemma}\label{l:morder}
Let $\Omega'\Subset\Omega$
and let $\rho_{m, \epsilon, \Omega'}(z)$ be the order
of $r\rightarrow M_{m, \epsilon, \Omega'}(z,r)$.
We then have
$$\rho_{m, \epsilon, \Omega'}(z)= \max(d-\epsilon,
\text{order of}\; n_{V(z),m}(r))$$
for $z\in \Omega'$, where, as above, $n_{V(z),m}(r)$ is the number of
resonances of $H_{V(z)}$ on $\Lambda_m$,
 $m \in \Z^*$  of norm at most $r>0$.
\end{lemma}
\begin{proof}
We first derive a convenient expression for the order of $n_{V(z),m}(r)$.
Recall that the zeros of $f_m(z,\lambda)$ with $\lambda \in \Lambda_0$
correspond, with multiplicity, to poles of $R_{V(z)}(\lambda)$ with
$\lambda \in \Lambda_m$ with at most a finite number of exceptions
for fixed $z$.
Using Lemma \ref{l:horder}, (\ref{eq:bd1}), and (\ref{eq:bd2}), we see that
if the order of $n_{V(z),m}(r)$ exceeds $d-1$, then it is equal to
\begin{equation}\label{eq:fmlimit}
\lim\sup_{r\rightarrow \infty}\frac{\log \int_0^{\pi}
\log |f_m(z,re^{i\theta})|d\theta}{\log r}.
\end{equation}  Conversely, if the limit in (\ref{eq:fmlimit}) exceeds
$d-1$, then it is equal to the order of $n_{V(z),m}(r)$.
Additionally, we note that by Lemma \ref{l:fmint}, if the limit in
(\ref{eq:fmlimit}) exceeds $d-1$, then it is equal to
\beq\label{eq:fmorder1}
\lim\sup_{r\rightarrow \infty}\frac{\log \left\{ \displaystyle{ \max_{|u|=r,\; |\Im u|\leq 1,\;
|\arg u|\leq \pi /4}\int_0^{\pi}
\log |f_m(z,ue^{i\theta})|d\theta} \right\} }{\log r}.
\eeq
We now relate the order of $n_{V(z),m}(r)$, given by (\ref{eq:fmorder1}) provided
it is greater than $d-1$, to the order of $M_{m,\epsilon,\Omega'}(z,u)$. When $r$ is very large,
\beq\label{eq:gorder1}
\log|\exp(u^{d-\epsilon})|\approx r^{d-\epsilon},\;
~\text{for}\; |u|=r,\; |\Im u|\leq 1,\;
|\arg u|\leq \pi /4 .
\eeq
We define the constant $r_{m,V, \Omega'} \equiv r_m(\langle V_{M,\Omega'}
\rangle^{1/(1-\epsilon)}+1)$
where $r_m$ is defined in Lemma \ref{l:wheremaxgm}.
Using (\ref{eq:gorder1}), this constant $r_{m,V,\Omega'}$, and the definitions
of $g_{m,\epsilon}$ and $M_{m,\epsilon, \Omega'}$,
we obtain
\begin{align*}
\lim \sup _{r\rightarrow \infty}\frac{\log M_{m,\epsilon, \Omega'}(z,r)}
{\log r}
& = \lim \sup _{r\rightarrow \infty}
\frac{\log \left\{ \displaystyle{ \max_{ r_{m,V, \Omega'} \leq |u| \leq r, \; |\Im u |\leq 1, \; |\arg u|\leq \pi/4}
g_{m,\epsilon}(z,u) } \right\} }{\log r}
\\ & = \max (d-\epsilon, \; \text{order of }\; n_{V(z),m}(r)).
\end{align*}
\end{proof}

\noindent
We now come to the main result of this section.

\begin{thm}\label{main-psh1}
Let $\Omega\subset \Complex^{\cd}$ be an open connected
set, let $m\in \Integers$, and let $V(z,x)$ satisfy the assumptions
(V).  If for some $z_m\in \Omega$, the function $n_{V(z_m),m}(r)$ has order
$d$, then there is a pluripolar set $E_m\subset \Omega$
such that $n_{V(z),m}(r)$ has order $d$ for $z\in \Omega \setminus E_m$.
Moreover, if for each $m\in \Integers^*$, there is a $z_m$ such that
$n_{V(z_m),m}(r)$ has order
$d$, then there is a pluripolar set $E$ such that for every $m\in \Integers^*$,
the function $n_{V(z),m}(r)$ has order $d$ for $z\in \Omega \setminus E$.
\end{thm}
The proof of this theorem uses the function $M_{m,\epsilon, \Omega'}$ which
we have developed.  Using this function, what remains to prove resembles
the proof of \cite[Corollary 1.42]{l-g}; see also
\cite[Proposition 2.3]{polar}.  We include the proof for the convenience of
the reader.
\begin{proof}
Let $\Omega'\Subset \Omega$ be an open connected set, with $z_m\in \Omega'$,
and let $0<\epsilon<1$.
By \cite[Proposition 1.40]{l-g},
since $M_{m,\epsilon, \Omega'}$ is
plurisubharmonic on $\Omega'\times \Complex$,
there is a sequence $\{ \psi_q\}$
of negative plurisubharmonic functions on $\Omega'$ so that
$$-(\rho_{m, \epsilon, \Omega'}(z))^{-1}=\lim \sup _{q\rightarrow \infty}
\psi_q(z)\; \text{for} \; z\in \Omega'.$$
By results of \cite{[Vodev1],[Vodev2]}
and Lemma \ref{l:morder}, the order $\rho_{m, \epsilon, \Omega'}(z)$
of $M_{m,\epsilon, \Omega'}$
does not exceed $d$.  Thus
$$\lim \sup _{q\rightarrow \infty}
\left(\psi_q(z)+\frac{1}{d}\right)\leq 0$$
and $$\lim \sup _{q\rightarrow \infty}
\left(\psi_q(z_m)+\frac{1}{d}\right)= 0.$$
Thus by \cite[Propostion 1.39]{l-g},
there is a pluripolar set $E'_m\subset \Omega'$ such that
$\rho_{m, \epsilon, \Omega'}(z) = d$ for $z\in \Omega'\setminus E'_m$.
Thus, by Lemma \ref{l:morder}, the order of $n_{V(z),m}(r)$ is
$d$ for $z\in \Omega'\setminus E'_m$.

Now let $\{ \Omega_j ~|~ j \in  \N \}$ be a countable collection of subsets $\Omega_j \Subset \Omega$,
$\cup_{j \in \N} \Omega_j = \Omega$, with $z_m \in \Omega_1$.
Then any $z \in \Omega$ belongs to $\Omega_j$, for some $j$, so $z_m$
and $z$ belong to $\Omega_j$. Applying the analysis above to $\Omega_j$, there is a pluripolar
set $E_{m,j} \subset \Omega_j$ so that the order of $n_{V(z),m}(r)$ is
$d$ for $z\in \Omega_j \setminus E_{m,j}$. We define $E_m = \cup_{j \in \N} E_{m,j}
\subset \Omega$. Then $E_m$ is pluripolar since it is the countable union of
pluripolar sets \cite[Proposition 1.37]{l-g}
and the order of $n_{V(z),m}(r)$ is
$d$ for $z\in \Omega \setminus E_m$.

Suppose that for each $m\in\Integers^*$ there is a $z_m$ with
$n_{V(z_m),m}(r)$ having order $d$.
Set $E=\cup_{\Z^*} E_m$, where $E_m$ is a pluripolar set such that
$n_{V(z),m}(r)$ has order $d$ for $z\in \Omega\setminus E_m$, as
guaranteed by the first part of the theorem.  Then
for all $m\in\Integers,$ $n_{V(z),m}(r)$ has order $d$ for
$z\in \Omega \setminus E$ and $E$ is pluripolar \cite[Proposition 1.37]{l-g}.
\end{proof}



\section{Proof of Theorem \ref{main1}}\label{mainth1}

In this section, we sketch the proof of Theorem \ref{mainth1}.
The main ingredients are Theorem \ref{main-psh1},
the lower bound on the number of resonances for the step potential proved in section \ref{examples},
and the argument of \cite{ch-hi1}.
The function $f_m (\lambda)$, defined in (\ref{eq:fmfirst}) for $m \in \Z^*$,
is analytic in
a neighborhood of
$\{ \lambda \in \Lambda ~|~  0 \leq \arg \lambda \leq \pi, ~\mbox{and} ~
|\lambda|> c_0\langle \|V\|_{L^{\infty}} \rangle \}$, where $c_0 > 0$
is the constant in Lemma \ref{l:basics}.
To indicate the dependence on
a particular potential $V$, we will write $f_{V,m}(\lambda)$.  However,
we use a fixed $\chi=\chi_V$, with $\chi\equiv 1$ on $K$.
For positive constants $N,M,q > 0$ and $j > 2 N c_0$, we define subsets of
$L^\infty_0 ( K; F)$ by
\bea\label{sets1}
A_m (N,M,q,j) & \equiv & \left\{ V \in L^\infty_0 (K; F) ~:~ \langle \| V \|_{L^\infty}
\rangle \leq N, \right. \nonumber \\
 & & \int_0^\pi ~\log | f_{V,m} ( r e^{i \theta} ) | ~d \theta
\leq M  r ^q, \nonumber \\
 & & \left. \mbox{for}\;  2N c_0 \leq r \leq j \right\} .
\eea

\begin{lemma}\label{closed1}
The set $A_m (N,M,q,j) \subset L_0^\infty (K;F)$ is closed.
\end{lemma}

\begin{proof}
Let $V_k\in A_m (N,M, q, j)$, such that $V_k\rightarrow V$ in the $L^{\infty}$ norm.
Then clearly $\langle \|V\|_{L^\infty}\rangle \leq N$.
We shall use (\ref{eq:fmfirst})
and the bound
\begin{equation}\label{eq:detbd}
|\det(I+A)-\det(I+B)|\leq \| A-B\|_1e^{\|A\|_1+\|B\|_1+1},
\end{equation}
see, for example, \cite{[Simon2]}.
With $A = (I + V_j R_0 (\lambda) \chi)^{-1} V_j T(\lambda) \chi$ and
$B= (I + V R_0 (\lambda) \chi)^{-1} V T(\lambda) \chi$, we
have
\beq\label{trace1}
\| A - B \|_1 \leq S_1+S_2,
\eeq
where
\beq\label{trace2}
S_1 = \left\|
\left\{ (I+V_j R_0 (\lambda) \chi)^{-1}(V-V_j) \chi R_0 (\lambda)
\chi (I+ VR_0 (\lambda) \chi)^{-1}
\right\} V \chi T(\lambda) \chi \right \|_1 ,
\eeq
and
\beq\label{trace3}
S_2 = \left \| ( I + V_j R_0 (\lambda) \chi )^{-1} (V_j - V) \chi T(\lambda) \chi
\right\|_1 .
\eeq
We note that when  $0\leq \arg \lambda \leq \pi$, and $
2N c_0 \leq |\lambda| \leq j$,
the term $( 1 + V_j R_0 (\lambda) \chi )^{-1}$ is uniformly bounded by (\ref{eq:invbound}).
Furthermore, it converges to
$( 1 + V R_0 (\lambda) \chi )^{-1}$ in operator norm as
$V_j \rightarrow V$ and we have the bound (\ref{eq:R0bound}).
Consequently, by the trace bound in Lemma \ref{l:fmbasic}, the
terms $S_1$ and $S_2$ in (\ref{trace2})-(\ref{trace3})
converge to zero.
Since the individual trace norms are uniformly bounded, it follows from (\ref{eq:detbd})
that $f_{V_k,m}(\lambda) \rightarrow f_{V,m}(\lambda)$ uniformly. Consequently,
$$ \int_0^\pi \log | f_{V,m}(r e^{i \theta})| ~d \theta \leq M  r^q , $$
for $r$ in the specified region,
so that $V \in A_m (N,M,q,j)$.
\end{proof}

In the next step, we characterize
those $V \in L^\infty_0 (K; F)$ for which the resonance counting function exponent
is strictly less than the dimension $d$.
For $N,\;M,\;q>0$, let
\beq
B_m (N,M,q)= \bigcap_{j\geq 2N c_0}A_m(N,M,q,j).
\eeq
Note that $B_m(N,M,q)$ is closed by Lemma \ref{closed1}.

\begin{lemma}\label{l:inB}
Let $V\in L^{\infty}_0 (K;F)$, with
$$\limsup _{r\rightarrow
\infty}\frac{\log n_{V,m} (r)}{\log r}<d.
$$
Then there exist $N,\; M\in \Natural$,
$l\in \Natural$, such that $V\in B_m (N,M,d-1/l)$.
\end{lemma}

\begin{proof}
Applying Lemma \ref{l:horder} with
$h = f_{V,m}$, we find
$$
\limsup_{r\rightarrow \infty}\frac{\log \int_0^\pi
\log| f_{V,m} (   re^{i\theta})|
~d \theta }{\log r}
=p<d.
$$
It follows that there is a $p'\geq p$, $p'<d$,
and an $M\in \Natural $ such that
$$
\int_0^\pi \log | f_{V,m} (r e^{i \theta})| ~d \theta \leq M r^{p'}
$$
when $r \geq c_0 \langle \|V\|_{L^\infty} \rangle$.  Choose $l\in \Natural $
so that $p'\leq d-1/l$ and $N\in \Natural$ so that $N\geq \|V\|_{L^\infty}$,
and then  $V\in B_m(N,M,d-1/l)$ as desired.
\end{proof}

\begin{lemma}
\label{l:gd}
The set
$$
\mathcal{M}_m = \left\{ V\in L^{\infty}_0(K;F): \;
\lim \sup_{r\rightarrow \infty}\frac{\log n_{V,m}(r)}{\log r}=d \right\}
$$
is a $G_{\delta}$ set, and thus
$$
\mathcal{M}=\bigcap_{m\in \Z^*}\mathcal{M}_m
$$
is as well.
\end{lemma}

\begin{proof} By Lemma \ref{l:inB}, the complement of
$\mathcal{M}_m$ is contained in
\begin{equation*}
\bigcup_{(N,M,l) \in \Natural^3}  B_m (N,M,d-1/l),
\end{equation*}
which is an $F_{\sigma}$ set since it is a countable union of
closed sets.  By Lemma \ref{l:horder}, if $V\in \mathcal{M}_m,$
then $V\not \in B_m(N,M,d-1/l)$ for any $N,\; M,\; l \in \Natural.$
Thus $\mathcal{M}_m$ is the complement of an $F_{\sigma}$ set.
\end{proof}

\noindent
We can now prove our theorem.

\begin{proof}[Proof of Theorem \ref{main1}]
Since Lemma \ref{l:gd} shows that $\mathcal{M}_m$ is a $G_{\delta}$ set,
we need only show that $\mathcal{M}_m$
 is dense in $L^{\infty}_0 (K;F)$.  To do this,
we follow the proof of \cite[Corollary 1.3]{polar} with appropriate modifications.
We give the proof here for the convenience of the reader.
Let $V_0\in L^{\infty}_0 (K;F)$ and let $\epsilon >0$.
By Theorem \ref{main111}, proved in section \ref{examples},
we may choose a nonzero, real-valued, spherically symmetric $V_1\in
L^{\infty}_0 (K;\Real)$ so that $V_1\in \mathcal{M}_m$, for all $m \in \Z^*$.
 Then consider the function $V(z)=V(z,x)=zV_1(x)+(1-z)V_0(x)$.
This satisfies the conditions of Theorem \ref{main-psh1}, with $V(1) = V_1$ and $V(0)=V_0$.
Thus, there exists a pluripolar set $E\subset \Complex$,
so that for $z\in \Complex\setminus E$, we have
$$
\limsup_{r\rightarrow \infty}\frac{\log n_{V(z),m}(r)}{\log r}=d.
$$
for $z\in \Complex
\setminus E$.
Since ${E}_{\restrict \Real}\subset \Real$ has
Lebesgue measure $0$
(e.g. \cite[Section 12.2]{ransford}),
we may find $z_0\in \Real$, $z_0\not \in E$, with
$|z_0|<\epsilon / (1+ \|V_0\|_{L^{\infty}}+\|V_1\|_{L^{\infty}})$.
Then $V(z_0)\in \mathcal{M}_m$
for all $m \in \Z^*$,
and $\|V(z_0)-V_0\|_{L^{\infty}}<\epsilon$.  Moreover, if $V_0$ is real-valued,
so is $V(z_0)$.
\end{proof}


\section{Lower bounds on Resonances for Certain Bounded,
Compactly-Supported Potentials in Even Dimensions}
\label{examples}

We prove Theorem \ref{main111} in this section.
That is, we study in even dimensions
the Schr\"odinger operator with a potential
given by a multiple of the characteristic function of the unit ball,
$V(x) = V_0 \chi_{[0,1]} (|x|)$, with $V_0 > 0$. We show that
$n_{V,m}(r)\geq C_{d,m}r^d+O(r^{d-1})$, with $C_{d,m}>0$. This
establishes the existence of potentials $V \in L^\infty_0
(\Real^d)$, for $d$ {\it even}, for which the resonance counting
function has the maximal order of growth $d$ {\it on each Riemann
sheet} $\Lambda_m$, defined by $m \pi < \arg \lambda < (m+1) \pi$, $m
\in \Z^*$. Similar calculations were performed in \cite{[Zworski2]} and \cite{stef06}
in the odd dimensional case.
In particular, Stefanov considered resonances for the transmission problem
with a stepwise constant wave speed in \cite[section 9]{stef06}. This consists in
studying the resonances of $- c(|x|)^2 \Delta$, with $c(|x|) = c \neq 1$, for $|x| \leq R$,
and one otherwise. Since the perturbation enters at
second-order, expansions to order $\nu^{-1}$ are sufficient.
For the zero-order perturbation by $V$ considered here,
more terms in the expansions are needed.


\subsection{Resonances for Schr\"odinger Operators on a Half-line}\label{resonances22}

The resonances for the spherically symmetric Schr\"odinger operator
$H_V = -\Delta + V$, with $V = V(|x|)$ a real-valued, bounded,
spherically symmetric potential, are completely determined by the
resonances for each one-dimensional Schr\"odinger operator $H_l$, $l
= 0, 1, 2, \ldots$, on $L^2 ( \Real^+)$ obtained by symmetry
reduction. We are interested in solutions to $H_\ell \psi_\ell = \lambda^2
\psi_\ell$, where
\bea\label{sphsymm1}
( H_\ell \psi_\ell )(r) & = & -
\psi_\ell''(r) - \frac{(d-1)}{r } \psi_\ell'(r) + \frac{l ( l + d -
2)}{r^2} \psi_\ell (r) + V(r) \psi_\ell (r)  \nonumber \\
 &=& \lambda^2 \psi_\ell(r).
\eea
Setting $\psi_\ell(r) = r^{-\frac{d-2}{2}} u_\ell(r)$, we find the
equation for $u_\ell$ is
\beq\label{sphsymm2}
\left[ \frac{d^2}{dr^2} +
\frac{1}{r} \frac{d}{dr} + \lambda^2 - \frac{L_d + \left( \frac{d-2}{2} \right)^2
}{r^2} - V(r) \right] u_\ell (r) = 0 ,
\eeq
where $L_d = \ell ( \ell
+ d - 2)$.

For the unperturbed case $V\equiv 0$, equation (\ref{sphsymm2}) is a
standard Bessel equation. We follow the notation of
\cite{[Olver1],[Watson1]} and define the Hankel functions as
\beq\label{defn-hankel1} H_s^{(1)} (z) = [ J_{-s}(z) - e^{-i s
\pi}J_s (z)]/ [ i \sin s \pi ], \eeq and
\beq\label{defn-hankel10}
H_s^{(2)} (z) = [ e^{is \pi} J_{s}(z) - J_{s} (z)]/ [ i \sin s \pi ],
\eeq
where these formulas are defined in the limit when $s \in \Z$. Note that $J_s
(z) = [ H_s^{(1)}(z) + H_s^{(2)}(z)] / 2$. A pair of linearly
independent solutions of (\ref{sphsymm2}) is \beq\label{solns1}
u_\ell^{(1)} (r) = J_{\ell + \frac{d-2}{2} } (\lambda r),
~~\mbox{and} ~u_\ell^{(2)} (r) = H_{\ell + \frac{d-2}{2} }^{(1)}
(\lambda r). \eeq The solution $u_\ell^{(1)}(r)$ is regular at
$r=0$, whereas the solution $u_\ell^{(2)}(r)$ satisfies the outgoing
condition, behaving like $\sim e^{i \lambda r}/ (\lambda r)^{1/2}$,
for $\lambda$ with $0 < \arg \lambda < \pi$,
as $r \rightarrow \infty$.
In light of these formulas, we
let $\nu \equiv \ell + (d-2)/2$, and note that it is an integer for
$d$ even, and half an odd integer for $d$ odd.
Returning to the free radial equation (\ref{sphsymm1}) with $V=0$, we
define spherical solutions as
\beq\label{linindsolns1}
j_\nu ( \lambda r) \equiv \left( \frac{\pi}{2r}
\right)^\frac{d-2}{2} J_{ \ell + \frac{d-2}{2} } (\lambda r) , ~~\mbox{and} ~~
h_\nu (\lambda r) = \left( \frac{\pi}{2r}
\right)^{\frac{d-2}{2}} H_{\ell + \frac{d-2}{2} }^{(1)} (\lambda r).
\eeq

For the Schr\"odinger operator $H_\ell$ on a half-line $\Real^+$ given in (\ref{sphsymm1}),
we define the resonances
of $H_\ell$
as poles in the meromorphic continuation
of the Green's function $G_\nu (r,r'; \lambda)$, with $\nu = \ell + (d-2)/2$. These poles
are independent of $r$ and $r'$ as they are the zeros of the continuation of a Wronskian
onto $\Lambda_m$, $m \in \Z^*$, or $\C^-$, for $d$ even or $d$ odd, respectively.
In the odd-dimensional case, these poles
correspond to those $\lambda \in \C^-$ for which there is
 a purely outgoing solution to
$H_\ell \psi_\nu = \lambda^2 \psi_\nu$ satisfying a Dirichlet boundary condition
at $r = 0$.

We now consider the potential $V(r) =
V_0 \chi_{[0,1]} (r)$, with $V_0 > 0$, and construct the Green's function
on the physical sheet $\Lambda_0$. Let $\Sigma (\lambda) \equiv ( \lambda^2 -
V_0)^{1/2}$, where the square root is defined so that this function has
branch cuts $( - \infty, -V_0^{1/2}] \cup [ V_0^{1/2}, \infty)$.
Because of the simple nature of the potential $V(r)$,
the equation for $0 < r < 1$ is
\beq\label{reduced1}
- \psi_\nu '' - \frac{(d-1)}{r} \psi_\nu ' + \frac{\ell ( \ell + d -2)}{r^2}
  \psi_\nu
  = \Sigma(\lambda)^2 \psi_\nu ,
\eeq
and for $r >1$, the solution $\psi_\nu$ satisfies the free equation
\beq\label{reduced2}
- \psi_\nu '' - \frac{(d-1)}{r} \psi_\nu ' + \frac{\ell ( \ell
+ d -2)}{r^2} \psi_\nu
  = \lambda^2 \psi_\nu.
\eeq
We choose two linearly independent solutions,
$\phi_{\nu}$ and $\psi_{\nu}$ of (\ref{reduced1})--(\ref{reduced2})
so that
$\phi_\nu ( r=0; \lambda) = 0$ and $\psi_\nu (r; \lambda) = h_\nu ( \lambda r)$
for $r >1$. The Green's function has the form
\beq\label{green1}
G_\nu (r,r'; \lambda) =
\frac{1}{W_\nu(\lambda) } \left\{ \begin{array}{ll}
     \phi_\nu ( r; \lambda  ) \psi_\nu ( r' ; \lambda ) , &  r < r' \\
    \phi_\nu ( r' ; \lambda ) \psi_\nu ( r; \lambda ) , &  r > r'
     \end{array}
     \right. ,
\eeq
where the Wronskian $W_\nu (\lambda)$, evaluated at $r=1$, is given by
\beq\label{wronskian1}
W_\nu (\lambda) = \Sigma (\lambda) j_{\nu} ' ( \Sigma (\lambda) ) h_{\nu }^{(1)} (\lambda ) -
\lambda j_\nu ( \Sigma ( \lambda)) h_{\nu}^{(1) '}  (\lambda ) .
\eeq
%
The function $W_\nu (\lambda)$ admits an analytic
continuation in $\lambda$ to $\C$, for $d$ odd, and to $\Lambda$, for $d$ even provided $\nu$
is also even. For the case of $d$ even and $\nu$ odd,
 $\Sigma (\lambda)W(\lambda)$ has an analytic continuation
to $\Lambda$.
This is consistent with the fact that the Green's function extends meromorphically
to $\Lambda$. For $d=2$, the Green's function has a logarithmic singularity at $z=0$.

We conclude that a value $\lambda_0 \in \Lambda_m$, $m \neq 0$, for $d$ even
(or $\lambda_0 \in \Complex^-$ if $d$ is odd) is a resonance if it satisfies the
following condition:
\beq\label{basic1}
\Sigma (\lambda_0) j_{\nu} ' ( \Sigma (\lambda_0) ) h_{\nu }^{(1)} (\lambda_0 ) =
\lambda_0 j_\nu ( \Sigma ( \lambda_0))
h_{\nu}^{(1) '}  (\lambda_0 ) .
\eeq
Using the definitions (\ref{linindsolns1}), 
we rewrite this relation as
\beq\label{basic11}
\Sigma (\lambda_0) J_{\nu} ' ( \Sigma (\lambda_0) ) H_{\nu}^{(1)} (\lambda_0 ) -
\lambda_0 J_\nu ( \Sigma (\lambda_0))
H_{\nu}^{(1) '}  (\lambda_0 ) = 0, ~\nu = \ell + (d-2)/2 .
\eeq

In order to study the defining equation (\ref{basic11}) on $\Lambda_m$, we define a
function $F_m^{(\nu)} (\lambda)$ on $\Lambda_0$ by
\beq\label{basic2}
F_m^{(\nu)}(\lambda) = \Sigma ( \lambda) J_{\nu }' ( \Sigma ( \lambda) ) H_{\nu }^{(1)} (e^{im \pi} \lambda )
 - e^{im \pi} \lambda J_{\nu} ( \Sigma  ( \lambda ))
H_{\nu}^{(1) '} ( e^{im \pi} \lambda  ) ,
\eeq
using the fact that $\Sigma ( e^{i m \pi} \lambda ) = \Sigma (\lambda)$, for $m \in \Z$.
It follows from
the fundamental equation (\ref{basic11})
that the zeros of $F_m^{(\nu)} (\lambda)$
on $\Lambda_0$ correspond to
the resonances of the one-dimensional Schr\"odinger operator
$H_\ell$ on the sheet $\Lambda_m$, for $|m| \geq 1$.
We will study the zeros of $F_m^{(\nu)} (\lambda)$
on $\Lambda_0$ in sections \ref{mthsheetlb1} and \ref{unifasymp2} below.

We make some important observations concerning $F_m^{(\nu)} (\lambda)$.
First, we note that when $V_0=0$,
$F_m^{(\nu)} (\lambda)$ is $e^{im \pi } \lambda$ times the Wronskian
of $J_\nu$ and $H_\nu^{(1)}$ evaluated at $e^{im \pi}\lambda$.
This is easily seen to be equal to the constant $2i \pi$, which
is consistent with the fact that there are no resonances in the free case.
Secondly, when $m=0$, corresponding to the physical sheet,
there are no zeros of $F_0^{(\nu)}(\lambda)$ on $\Lambda_0$, see (\ref{mzeros01b}).
Thirdly, this equation reflects the symmetry properties of the meromorphic
continuation of the resolvent and $S$-matrix depending on whether $d$ is odd or even.
As mentioned in section \ref{preliminaries1}, for $d \geq 2$ even, and $V$ real,
if $k \in \Lambda_m$, $m \in \Z^*$
is a zero, then so is $- \overline{k} \in \Lambda_{-m}$.
In the odd-dimensional case, again with $V$ real,
if $k \in \C^-$ is a zero, then so is $- \overline{k} \in \C^-$,
so the resonances are symmetric about the imaginary axis.

\vspace{.1in}
\noindent
{\bf Remark.}
For $d =3$, the poles of the $S$-matrix associated
with the Schr\"odinger operator (\ref{sphsymm1}) with a spherically symmetric
step potential barrier $V_0 > 0$ (the case considered here) or well $V_0 < 0$, were
studied by Nussenzveig \cite{[Nussenzveig1]}. The $S$-matrix is explicitly computable.
Let $\Sigma( \lambda) \equiv ( \lambda^2 \pm V_0)^{1/2}$,
where the plus sign is for the potential well, and
the minus sign is for the barrier.
For each angular momentum $\ell \geq 0$, the corresponding component
of the $S$-matrix is
\beq\label{s-matrix1}
S_\ell ( \lambda) = - \frac{ \lambda j_\ell ( \Sigma (\lambda )){ h_\ell^{(2)}} ' (\lambda)
- \Sigma (\lambda) j_\ell ' (\Sigma (\lambda)) h_\ell^{(2)} (\lambda) }{  \lambda j_\ell ( \Sigma (\lambda ) )
 {h_\ell^{(1)} }' (\lambda) - \Sigma (\lambda) {j_\ell}' (\Sigma (\lambda)) h_\ell^{(1)} (\lambda)}.
\eeq
This should be compared with (\ref{basic11}). As this formula also holds in the even-dimensional
case, we see that the resonances defined via the continuation of the
resolvent are the same as those defined by the
continuation of the $S$-matrix.  Nussenzveig studied the behavior
of bound states $V_0 < 0$ and resonances  $V_0 > 0$, describing various asymptotics expansions
of their real and imaginary parts,
especially in the case of $\ell = 0$.


\subsection{Analysis of Zeros: an Overview}\label{sec:summary-zero1}

We discuss here the strategy for estimating the zeros of $F_m^{(\nu)}(\lambda)$
on the sheet $\Lambda_0$. We prove that inside a semicircle of radius $r > 0$,
the angular momentum states with $\ell < r$ contribute at least $c_mr^d$ zeros, where $c_m > 0$
depends on $m \in \Z^*$. We need to control the
remainder term
uniformly in $r$ and $\ell$ as $r \rightarrow \infty$.
Hence, we need to study the zeros of $F_m^{(\nu)}(\lambda)$
as $|\lambda|$ and $\nu$ go to infinity.
We define new variables $z \equiv \lambda / \nu$
and $\tilde{z} (z) \equiv (z^2 - \nu^{-2} V_0)^{1/2}$,
so that we study
\beq\label{basic02}
F_m^{(\nu)}(\nu z)  = \nu \tilde{z} (z) J_{\nu }' ( \nu \tilde{z} ) H_{\nu }^{(1)} (e^{im \pi} \nu z ) -
e^{im \pi} \nu z J_{\nu} ( \nu \tilde{z}(z))
H_{\nu}^{(1) '} ( e^{im \pi} \nu z  ) .
\eeq
for $z$ in a fixed region. For this, we use the
uniform large-order asymptotics of the Bessel and Hankel functions
due to Olver \cite{[Olver2],[Olver3]}.
A special role in these uniform asymptotics is played by the compact, eye-shaped region $K$ in the complex
plane defined as follows. This region is bounded by the curves containing the points labeled
$B,P,E,E',P'$ in Figure 1 and \cite[chapter 11]{[Olver1]}. Let $t_0$
be the positive root of $t = \coth t$, so $t_0 \sim 1.19967864 \ldots$. The region $K$ is
the symmetric region in the neighborhood of the origin bounded in $\C^+$ by the curve
\beq\label{k-region1}
z = \pm ( t \coth t - t^2)^{1/2} + i ( t^2 - t \tanh t)^{1/2}, ~0 \leq t \leq t_0,
\eeq
intercepting the real axis at $\pm 1$ and intercepting the imaginary
axis at $i z_0$, where $z_0 = (t_0^2 - 1)^{1/2} \sim 0.66274 \ldots$. The region $K$
is bounded by the conjugate curve in the lower half-plane.
The significance of this region $K$ is illustrated by the fact
that the ordinary Bessel function $J_\nu (z)$ decays
exponentially in $\nu$ for $z \in Int (K)$, whereas it grows
exponentially in $\nu$ and $z$ for $z \in Ext (K)$.

Essential to understanding the uniform asymptotics of the Bessel and Hankel functions is the mapping
on the complex plane $z \rightarrow \zeta (z)$ defined by
\beq\label{map1}
\rho (z) \equiv \frac{2}{3} \zeta^{3/2} (z) = \log \frac{1 + \sqrt{1 - z^2} }{z} - \sqrt{1 - z^2} .
\eeq
The relationship between the variables $z$,
$\zeta (z)$, and $\rho (z)$ is clarified by Figure 1 and is
described in detail beginning on page 335 of Olver's article \cite{[Olver3]}
(reproduced on page 420 of his book
\cite{[Olver1]} where the author changed notation and used $\xi$ for what is
called $\rho$ here and in \cite{[Olver3]}).
We present these figures, with permission of the author and publisher,
in Figure 1.


\begin{figure}
  \centering
\begin{center}
\includegraphics[height=40mm,angle=180]{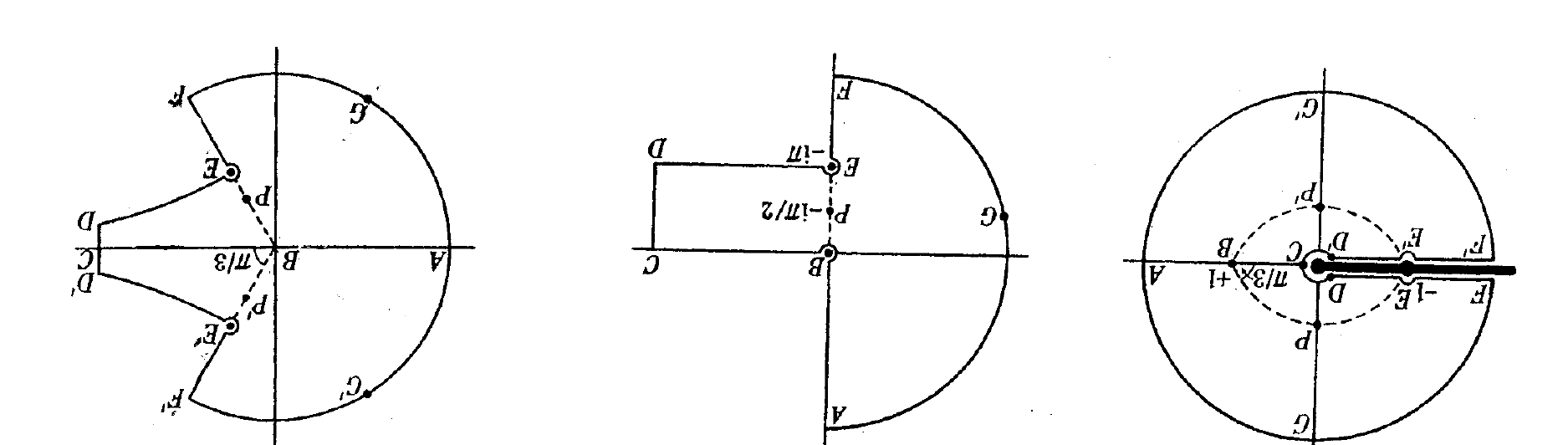}
\end{center}
...
  \caption{Left: the $z$-plane with the region $K$ enclosed by the curves joining $B,P,E,E', P'$.
  Middle: The image of the $z$ plane under the mapping $\rho$  defined in (\ref{map1}).
  Right: The image of the $z$-plane under the mapping $\zeta$ defined in (\ref{map1}).
Reproduced, with permission of the author and publisher, from \cite[page 336]{[Olver3]}.}
\end{figure}


The cut along the negative real axis in the $z$-plane is mapped onto two
cuts along the rays $\arg \zeta = \pm \pi / 3$. The eye-shaped region $K$ is mapped
into the sector $| \arg \zeta | < \pi / 3$. A neighborhood of the positive real axis in
the $z$-plane is mapped onto a neighborhood of the negative real axis in the $\zeta$-plane.
The upper boundary of the eye-shaped region $K$ in the $z$-plane is mapped onto
the line segment $(0, -i \pi)$ in the $\rho$-plane, with the point $ z_0 i$ being mapped to $-i \pi / 2$.
We analyze the function $F_m^{(\nu)}(\nu z )$ for large $\nu$ with $z$ in a neighborhood of $K$.

\subsection{Resonances on the $m^{th}$ Sheet $\Lambda_m, ~m \in \Z^*$}\label{mthsheetlb1}

We consider the general case $m \neq 0$ and prove a lower bound on the number of resonances.
The zeros of $F_m^{(\nu)}(\nu z)$, defined in (\ref{basic02}) for
$z \in \Lambda_0$, correspond to the resonances of $H_\ell$ on the sheet
$\Lambda_m$. We recall from section \ref{sec:summary-zero1} that in order to use the uniform asymptotics
of section \ref{unifasymp2},
we defined $z \equiv \lambda / \nu$ and $\tilde{z} (z) \equiv (z^2 - \nu^{-2} V_0)^{1/2}$.
We also recall that $d \geq 2$ is even so that $\nu = \ell +(d-2)/2$ is a
nonnegative integer.
Using formulas (\ref{hankel11}) and (\ref{hankel-der1}), we find
that
\beq\label{m-sheet1}
F_m^{(\nu)} (\nu z)  = (-1)^{m\nu} [ F_0^{(\nu)}(\nu z ) - 2 m G_0^{(\nu)} (\nu z) ],
\eeq
where, from (\ref{basic02}),
\beq\label{zero-sheet1}
F_0^{(\nu)} (\nu z) = \nu \tilde{z} J_\nu ' (\nu \tilde{z}) H_{\nu}^{(1)} (\nu z) -
 \nu z J_\nu (\nu \tilde{z}) H_{\nu}^{(1)'} (\nu z ) ,
\eeq
and we define
\beq\label{m-sheet21}
G_0^{(\nu)} ( \nu z) \equiv  \nu \tilde{z} J_\nu '
(\nu \tilde{z})  J_{\nu}  (\nu z )  - \nu z J_{\nu} (\nu \tilde{z} ) J_\nu^\prime (\nu z ) .
\eeq
It follows that in order to study the solutions of $F_m^{(\nu)} ( \nu z ) = 0$,
we need to consider those $z$ with $0 < \arg z < \pi$ for which
\beq\label{m-sheet2}
{F}_0^{(\nu)} ( \nu z) = 2 m G_0^{(\nu)} ( \nu z) .
\eeq
It is sufficient for the lower bound to prove that
for any
$\nu < r$, for
$r >> 0$ sufficiently large, that there are at least $\nu(1-\epsilon_1)$,
$\epsilon_1>0$, solutions
of this equation in the half-disk
$\Im \lambda > 0$ and $| \lambda | \leq r$, uniformly in $r$ and $\nu$.

To prove that there are at least $\nu(1-\epsilon_1)$ zeros
of $F_m^{(\nu)}(\lambda)$ near the upper boundary of the
eye-shaped region $\nu K$,
we concentrate a small region $\Omega_{1, \epsilon}$ near the upper
boundary of $K$. In particular, we define, for fixed $\epsilon>0$,
\begin{equation}\label{eq:omega1}
\Omega_{1, \epsilon}= \{ z\in \C^{+}:
\text{dist}\; (z,\partial K^+)<\epsilon\} \cap\{ z\in \C^{+}:\;
|z+ 1|>\epsilon\; \text{and}\; |z-1|>\epsilon\}.
\end{equation}
We recall that $\rho (z)$ is defined in (\ref{map1})
and that the region $\Omega_{1, \epsilon}$ near the upper boundary of $K$
is mapped onto a neighborhood of the line segment $(-i \pi
+ih(\epsilon)  , -ih(\epsilon))$, where $h(\epsilon)>0$,
$h(\epsilon)=\mathcal{O}(\epsilon)$,
in the $\rho$-plane
(see the middle figure of Figure 1 or Fig.\ 3--6 in \cite{[Olver2]}).
Consequently, we will identify the
zeros near the upper edge
of $K$ in the $z$-plane with their image $\rho (z)$ near this line segment in the $\rho$-plane.

We compute the uniform asymptotics of $F_0^{(\nu)} (\nu z)$ in section
\ref{subsec:condition-m-0}.
From (\ref{mzeros01a}), we have
\bea\label{mzeros01b}
F_0^{(\nu)} (\nu z) &=& \frac{- 2i}{\pi} \left\{ 1 - \frac{1}{\nu} \left[ \frac{V_0 (1-z^2)^{1/2}}{2 z^2}
\right]
+\mco\left(\frac{1}{\nu^2}\right) \right\}.
\eea
The uniform asymptotics of the term $G_0^{(\nu)}$ on the right of (\ref{m-sheet2})
is also computed in section \ref{subsec:condition-m-0}:
\beq\label{m-sheet3}
G_0^{(\nu)} (\nu z) =  \frac{e^{-2 \nu \rho}}{2 \pi} \left[ \frac{V_0 }{2 \nu^2 (1-z^2)} + \mathcal{O}
\left( \frac{1}{\nu^3} \right) \right] ,
\eeq
where the error is uniform for
$z\in \Omega_{1,\epsilon}$.
Consequently, the condition for zeros on the $m^{th}$-first sheet is
\beq\label{reson-mth}
e^{2 \nu \rho (z) } \left( 1 + g_1(z, \nu)  \right)
= \frac{i m  V_0 }{4 \nu^2} \left( \frac{1}{1-z^2} \right) +
g_2(z, \nu),
\eeq
where $g_1(z,\nu)= \mathcal{O}(1/\nu)$, and $g_2(z,\nu)= \mathcal{O}(1/\nu^3)$, both uniformly for
$z\in \Omega_{1,\epsilon}$.
We note that for $V_0= 0$ there are no solutions to this equation.

We now study the solutions of
(\ref{reson-mth}).
The variable $\rho$ lies in a
set that is the image of $\Omega_{1, \epsilon}$ under the mapping
$z \rightarrow \rho$ given in (\ref{map1}). This set
contains a neighborhood of an interval of the negative
imaginary axis of the form
 $(-\pi+ h(\epsilon), -h(\epsilon))i
\subset (-\pi,0)i$. We will prove that there exists at least $\nu (1- \epsilon_1)$ solutions in
a neighborhood of this set.
We rewrite (\ref{reson-mth}) as
\beq\label{reson-second}
\nu^2 e^{2 \nu \rho} (1-z^2) \left( 1 + g_1(z, \nu)  \right)
 = \frac{i m V_0 }{4}
  +g_3(\rho,\nu),
\eeq
with $g_3 (z, \nu) = \mathcal{O} (\nu^{-1})$, uniformly for $\rho$ in the image
of $\Omega_{1,\epsilon}$
We define two functions:
\beq\label{eq:gfnc1}
g(z , \nu ) = \nu^2 e^{2 \nu \rho} - \frac{i m V_0}{4} ,
\eeq
and
\beq\label{eq:ffnc1}
f(z, \nu) = \nu^2 (1-z^2) e^{2 \nu \rho} \left( 1 + g_1(z, \nu)  \right)
- \frac{i m V_0}{4} - g_3(z, \nu).
\eeq
We will prove that near each zero of $g(z, \nu)$, in a neighborhood of  $(-\pi+ h(\epsilon), -h(\epsilon))i
\subset (-\pi,0)i$,
there is a zero of $f(z, \nu)$ using
Rouche's Theorem.

\begin{lemma}\label{l:zerosofg}
The function $g( z, \nu )$, defined in (\ref{eq:gfnc1}) for $m \in \Z^*$, has infinitely-many zeros of the
form
\beq\label{g-zeros1}
\rho_k = \left\{ \frac{1}{2 \nu} \log \left( \frac{|m| V_0}{4} \right) - \frac{ \log \nu}{\nu} \right\}
+ i \frac{\pi}{\nu} \left[ k +
\operatorname{sgn}(m)\frac{1}{4} \right] , ~~ k \in \Z .
\eeq
\end{lemma}

\begin{proof} Viewing the function $g(z, \nu)$ as a function of $\rho$,
it is clear that it is periodic in $\Im \rho$ with period $\pi / \nu$. We can then explicitly solve
$g(z , \nu) = 0$.
\end{proof}

\noindent
The values $\rho_k$ provide the approximate solutions of $f(z,\rho) = 0$ as the next lemma shows.

\begin{lemma}\label{l:zerosoff}  For $m\in \Integers^*$,
in a neighborhood of each zero $\rho_k$ of $g(z, \nu)$ with imaginary part
in $( -i \pi + i2 h(\epsilon), -i2 h(\epsilon))$, there is exactly one zero
of $f(z,\nu)$. Consequently,
there are at least $\nu(1-\epsilon_1)$, $\epsilon_1=\mco(\epsilon) >0$
 zeros of $f(z, \nu)$ in a neighborhood of
the interval on the negative imaginary axis $( - \pi  , 0)i$ for all $\nu > 0$ large.
\end{lemma}
We remark that $\epsilon_1>0$ may be made arbitrarily
small by choosing $\epsilon>0$ sufficiently small.

\begin{proof}
\noindent
1. We use Rouche's Theorem applied on
 small rectangular contour $\mathcal{C}$ about $\rho_k$ with
vertical sides $A$ and $B$, and
horizontal sides $C$ and $D$, chosen so that exactly one zero $\rho_k$ of
$g(z, \nu)$ lies in the rectangle.
This is possible as the approximate solutions (\ref{g-zeros1}) are separated
as $|\Im \rho_k - \Im \rho_{k+1}
| = \pi / \nu$. In order to apply Rouche's Theorem, we must show that
on each segment of the contour $\mathcal{C}$ we have
\beq\label{eq:inequality1}
|f(z,\nu) - g(z, \nu)| < |g(z, \nu)|,
\eeq
for all $\nu$ large.
We will repeatedly use the fact that for $z \in \Omega_{1, \epsilon}$, we have $|z|^2 < (1 - \delta (\epsilon) )$,
for some small $\delta > 0$, a function of $\epsilon > 0$ used to define $\Omega_{1, \epsilon}$
specified in (\ref{eq:omega1}).

\noindent
2. Vertical sides $A$ and $B$. Let side $A$ lie on the negative imaginary axis
and be parameterized by $i t$, for $t < 0$. Along this side,
$|g(z,\nu)| = \nu^2 + \mathcal{O}(1)$
and $|(f(z,\nu) - g(z , \nu))\restrict_A| \leq (1- \delta) \nu^2 + \mathcal{O}( \nu^{-1} )$, so
we have (\ref{eq:inequality1}).
Along edge $B$, we let $\rho = - a + it$, for $a = \alpha \log \nu / \nu > 0$, with $\alpha > 1$, and $t<0$. We easily
find that $| (f(z, \nu) - g(z, \nu) )\restrict_B | \leq (1 - \delta) \nu^{- 2( \alpha - 1)} + \mathcal{O}( \nu^{-1})$.
On the other hand, we have $|g(z,\nu) \restrict_B| = (|m| V_0 / 4) ( 1 + \mathcal{O}( \nu^{-2 (\alpha - 1)} ) )$.
It follows that (\ref{eq:inequality1}) holds along $A$ and $B$.

\noindent
3. Horizontal sides $C$ and $D$. For zero $\rho_k$, with $k < 0$, we choose the
straight line segment $C$ so that on $C$
we have
\beq\label{eq:impart1}
\Im \rho_{k-1} = \frac{\pi}{\nu} \left[ k -1+\text{sgn}(m) \frac{1}{4} \right] < \Im \rho <  \Im \rho_{k} =
\frac{\pi}{\nu} \left[ k + \text{sgn}(m)\frac{1}{4} \right] ,
\eeq
and such that
\beq\label{eq:impart2}
2 \nu \Im \rho = \frac{\pi}{2} ~\mbox{mod} ~ 2 \pi .
\eeq
With $\Re \rho=t$, we then have
$$
| (f(z, \nu) - g(z, \nu) )\restrict_C |
\leq (1 - \delta) \nu^2 e^{2 \nu t}(1+\mco(\nu^{-1}))+\mco(\nu^{-1}),$$
whereas $|g(z, \nu)|_C = [ (\nu^2 e^{2 \nu t})^2 + (|m| V_0 / 4)^2 ]^{1/2},$ so we
again have (\ref{eq:inequality1}). A similar choice of $D$ insures the same inequality.
\end{proof}

\noindent
We now prove Theorem \ref{main111} for the case $m \in \Z^*$.
\begin{proof}
Recall that $\nu=l+(d-2)/2$, $l \in \Natural$ and $\epsilon_1>0$ is arbitrary.
From Lemma \ref{l:zerosoff} it follows that
there are, for each $\nu > 0$ large, at least $\nu(1-\epsilon_1)$ zeros
near the upper edge of the eye-shaped region $K$.
Furthermore, each zero of ${F}_m^{(\nu)}(\lambda)$ corresponds to a resonance of
multiplicity $m(l)$, the
dimension  of the space of spherical harmonics
on $\Sphere^{d-1}$ with eigenvalue $l(l+d-2)$.
Since $m(l)\geq cl^{d-2}+ \mathcal{O} (l^{d-3})$, for some $c>0$,
it follows from Lemma \ref{l:zerosoff}
that
\beq\label{count1}
n_{1,V}(r) \geq \sum_{\ell = 1}^{[r]}\frac{1}{4}(l-(d-2)/2)(cl^{d-2}+ \mathcal{O} (l^{d-3}))
\geq Cr^{d}+ \mathcal{O} (r^{d-1}),
\eeq
for some $C>0$, depending on $m \in \Z^*$.
\end{proof}
This proves the lower bound for the even-dimensional case on the $m^{th}$-sheet, $m \in \Z^*$. We recall
that for $V$ real, the symmetry
of the zeros means that the resonances on $\Lambda_{-m}$ are the same as those on $\Lambda_m$.

\vspace{.1in}
\noindent
{\bf Remark.}
We make some historic comments relevant to the odd dimensional case. For $d=3$,
resonances of spherically-symmetric, compactly-supported potentials were studied
by many physicists and Newton provided a nice summary \cite{[newton1]}. In particular,
R.\ Newton studied the zeros of the Jost function $f_\ell ( \lambda)$
for each angular momentum component in dimension three.
These are the same as the zeros of the function $F_1^{(\nu)} (\lambda)$.
For a real potential with compact support inside the ball of radius $R > 0$,
he gives a proof that $f_\ell (\lambda)$ has infinitely many complex roots in the lower-half complex plane,
and that these roots are symmetric with respect to the imaginary axis. This follows from
the fact that $f_\ell (\lambda) f_\ell (-\lambda)$ is an entire function of $\lambda^2$ of order $1/2$. It also
follows that only finitely-many roots lie on the negative imaginary axis. Finally,
he proves that if $\lambda_n$ is a sequence of roots with positive real parts, then
$\Re \lambda_n = n \pi / R + \mathcal{O} (1)$ and $\Im \lambda_n = [ ( \sigma +2)/ (2R)] \log n + \mathcal{O}(1)$.
In particular, this shows that there are infinitely many roots in the region $2 \pi - \epsilon
< \arg \lambda < 0$. These lie outside of the region considered above.

\section{Appendix: Analytic Continuation and Uniform Asymptotics of Bessel and Hankel Functions}
\label{app-bounds}

In this appendix, we provide all the details necessary for
obtaining a lower bound on the number of zeros of
the function $F_m^{(\nu)}(\lambda)$ on the $m^{th}$-sheet.
In the first section we give the analytic continuation of the Bessel and Hankel functions
following Olver \cite[chapter 7]{[Olver1]}.
We then summarize the uniform large-order asymptotics of the Bessel and Hankel functions proved by Olver
\cite{[Olver2],[Olver3]}.
These rely on the asymptotic expansion of the Airy functions (see, for example, \cite[appendix]{[Olver3]})
that we present in the next section. Finally, we
compute the uniform asymptotics of the terms occurring in $F_m^{(\nu)}(\lambda)$.


\subsection{Analytic Continuation of Bessel and Hankel Functions}\label{zerosbesshank1}

The analytic continuations of the ordinary Bessel functions $J_\nu (z)$,
for $\nu \in \R$,
from the region $\Lambda_0$ to the region $\Lambda_m$,
are obtained by the formula
\beq\label{bessel1}
J_{\nu }(z e^{i m \pi }) = e^{ i m \nu \pi} ~J_{\nu} (z).
\eeq
It follows that
\beq\label{bessel-der1}
J_{\nu }'(z e^{i m \pi }) = e^{ i m  \pi ( \nu + 1)} ~J_{ \nu}' (z).
\eeq
As for the Hankel function $H_\nu^{(1)} (z)$,
and $z \in \Lambda_0$, the analytic continuation to $\Lambda _m$, with $m
\in \Z^*$, is obtained
through the following formula (e.g.\  \cite[chapter 7]{[Olver1]}),
\beq\label{hankel1}
H_\nu^{(1)}( e^{i m \pi} z) = - \frac{ \sin (m-1) \nu \pi }{ \sin \nu \pi
} H_\nu^{(1)}(z ) - e^{- i \pi \nu } \frac{ \sin \nu m \pi }{ \sin \nu \pi }
H_\nu^{(2)}  (z) ,
\eeq
where, if $\nu \in Z$, we define the right side by the limit.

In our case, $\nu = \ell + (d-2) / 2$, with $\ell =  0, 1, 2, \ldots$, so that
when $d$ is even, $\nu$ is a non-negative integer
$\nu = 0, 1, 2, \ldots$, and when $d$ is odd, $\nu$ is
half an odd integer. In the case when $d \geq 4$ is even, the Hankel functions are
analytic on the Riemann surface of the logarithm $\Lambda$. When $d = 2$
there is a logarithmic singularity at the origin $z=0$.
In the case when $\nu \in \Z$, formula (\ref{hankel1}) becomes
\bea\label{hankel11}
H_\nu^{(1)}( e^{i m \pi} z) &= &(-1)^{m \nu + 1} [ (m-1) H_\nu^{(1)} (z) + m H_\nu^{(2)} (z) ] \nonumber \\
 &=& (-1)^{m \nu} [ H_\nu^{(1)} (z) - 2 m J_\nu (z) ] .
\eea

As for the derivatives of the Hankel function $H_\nu^{(1)} (z)$,
for $0 < \arg z < \pi$,
the analytic continuation to the sheet $\Lambda _m$ with $m \pi <
\arg z < (m+1) \pi$, for $m \in Z$, is obtained
from (\ref{hankel1}).
Restricting ourselves to the case of interest $\nu \in \Z$, we obtain:
\bea\label{hankel-der1}
{H_\nu^{(1)}}^\prime( e^{i m \pi} z)& = &(-1)^{m (\nu+1) + 1} [ (m-1) {H_\nu^{(1)}} ' (z) + m {H_\nu^{(2)}}'
(z) ] \nonumber \\
 &=&  (-1)^{m \nu + 1} [ {H_\nu^{(1)}}^\prime (z) - 2 m J_\nu^\prime (z) ]   .
\eea


\subsection{Asymptotic Expansions of Bessel and Hankel Functions}\label{JHasypmt}

The asymptotics of the Bessel and Hankel functions used here are expressed in
terms of Airy functions. As in \cite{[Olver1]}, we adopt the convention
that $Ai(w)$ has its zeros on the negative
real axis.
The index $\nu$ is real and positive.
It is convenient to define the following functions:
\beq\label{defn-phi1}
\phi ( \zeta) \equiv \left( \frac{ 4 \zeta }{ 1 - z^2} \right)^{1/4} = \left( - \frac{2}{z} \frac{dz}{d \zeta}
\right)^{1/2},
\eeq
and
\beq\label{defn-chi1}
\chi ( \zeta) \equiv \frac{ \phi' ( \zeta )}{ \phi (\zeta)} = \frac{4 - z^2 [ \phi (\zeta)]^6}{16 \zeta},
\eeq
and
\beq\label{defn-psi1}
\psi ( \zeta) \equiv \frac{2}{z \phi ( \zeta)}.
\eeq
We also need the following series expansions:
\beq\label{coef1}
F_1( \zeta, \nu ) = \sum_{j=0}^\infty \frac{ A_j ( \zeta ) }{ \nu^{2j}}, ~~F_2(\zeta , \nu) =
\sum_{j=0}^\infty \frac{ B_j ( \zeta ) }{ \nu^{2j}} ,
\eeq
with $A_0 ( \zeta) = 1$, and the remaining coefficients are determined
recursively, see \cite[section 9]{[Olver2]}.
The functions $G_j ( \zeta , \nu)$, $j=1,2$, are given as infinite series
\beq\label{hankel-coef0}
G_1( \zeta, \nu ) = \sum_{j=0}^\infty \frac{ C_j ( \zeta ) }{ \nu^{2j}}, ~~G_2(\zeta , \nu)
= \sum_{j=0}^\infty \frac{ D_j ( \zeta ) }{ \nu^{2j}} ,
\eeq
where the coefficients are
\beq\label{hankel-coef1}
C_j ( \zeta ) = \chi ( \zeta) A_j(\zeta) + A_j' ( \zeta)
+  \zeta B_j(\zeta) ,
\eeq
\beq\label{hankel-coef2}
D_j ( \zeta ) = \chi ( \zeta ) B_{j-1} ( \zeta) + B_{j-1} ' ( \zeta ) +  A_j ( \zeta ),
\eeq
with $B_{-1} ( \zeta ) = 0$ and $D_0 ( \zeta ) = 1$.

For the ordinary Bessel functions with $| \arg z | < \pi - \epsilon$, for any
$\epsilon > 0$, Olver proved that
\beq\label{besselasympt1}
J_\nu (\nu z) \sim \phi (\zeta)
\left( \frac{Ai( \nu^{2/3} \zeta )}{\nu^{1/3}} F_1 ( \zeta, \nu ) +
\frac{Ai '( \nu^{2/3} \zeta )}{\nu^{5/3}} F_2 ( \zeta, \nu ) \right) .
\eeq
He also proved that the asymptotic expansion for the derivatives can be obtained
by differentiation. It is useful to recall that the Airy function $Ai(w)$
(and $Ai_{\pm 1} (w)$ introduced below)
satisfies the differential equation $Ai'' ( u) = u Ai (u)$. For the
ordinary Bessel functions, one obtains
for $| \arg z | < \pi - \epsilon$, and for any $\epsilon > 0$,
\beq\label{bessel-deriv-asympt1}
J_\nu' (\nu z) \sim - \psi ( \zeta)
\left( \frac{Ai( \nu^{2/3} \zeta )}{\nu^{4/3}} G_1 ( \zeta, \nu ) +
\frac{Ai '( \nu^{2/3} \zeta )}{\nu^{2/3}} G_2 ( \zeta, \nu ) \right) .
\eeq

We also need the uniform asymptotics of the Hankel function.
As in Olver \cite{[Olver1]}, we define $Ai_{\pm 1} (w) = Ai ( e^{\mp 2 \pi i / 3} w )$.
For $ | \arg z| \leq \pi - \epsilon$, for any $\epsilon > 0$, we have
\beq\label{hankelasympt1}
H^{(1)}_\nu ( \nu z) \sim 2 e^{- i \pi /3} ~\phi ( \zeta)
~\left( \frac{Ai_{- 1}( \nu^{2/3} \zeta )}{\nu^{1/3}} F_1 ( \zeta, \nu )
+ \frac{Ai_{- 1}'( \nu^{2/3} \zeta )}{\nu^{5/3}} F_2 ( \zeta, \nu ) \right),
\eeq
where the functions $F_j ( \zeta , \nu)$, $j=1,2$ are given in (\ref{coef1}).
As with the Bessel functions, we can differentiate this expansion and obtain
\beq\label{hankel-deriv-asympt1}
{H^{(1)}_\nu} ' ( \nu z) \sim - 2 e^{-i \pi /3} ~\psi ( \zeta)
~\left( \frac{Ai_{ - 1}( \nu^{2/3} \zeta )}{\nu^{4/3}} G_1 ( \zeta, \nu )
+ \frac{Ai_{ - 1}'( \nu^{2/3} \zeta )}{\nu^{2/3}} G_2 ( \zeta, \nu ) \right).
\eeq
The functions $G_j ( \zeta , \nu)$, $j=1,2$, are given in (\ref{hankel-coef0}).
As above, these expansions are uniform in $ | \arg z| \leq \pi - \epsilon$, for any $\epsilon > 0$.


\subsection{Asymptotics for Airy Functions}\label{subsub-airy}

The asymptotics of the Bessel and Hankel functions are obtained
from the asymptotic expansions for the Airy functions (e.g.\ \cite[appendix]{[Olver3]}). These expansions
depend on whether the argument is near the zeros of the Airy functions along the
negative real axis, or away from them. For the Hankel functions of the first kind we only need
the asymptotics away from the zeros of the Airy functions. Let $\xi = (2/3) w^{3/2}$ where
$\xi$ is the principle value.
For $| \arg w | < \pi - \epsilon$, for any $\epsilon > 0$, the argument $w$ is away from the zeros
of the Airy function, and we have
\beq\label{airy1}
Ai (w) \sim \frac{e^{- \xi}}{ 2 \pi^{1/2} w^{1/4} } \hat{F}_1 ( \xi),
~\mbox{and} ~Ai' (w) \sim \frac{- w^{1/4}e^{- \xi}}{ 2 \pi^{1/2}  } \hat{G}_1 ( \xi).
\eeq
The functions appearing in these expansions are
\beq\label{airy-coef1}
\hat{F}_1 (\xi) \equiv \sum_{j=0}^\infty \frac{ c_j}{\xi^j }, ~\mbox{and} ~\hat{G}_1 (\xi)
\equiv  \sum_{j=0}^\infty \frac{ d_j}{\xi^j }  .
\eeq
The coefficients $c_j$ and $d_j$ are numerical constants given by
\beq\label{airy-coef2}
c_j = (-1)^j \frac{(2j+1)(2j+3) \cdots (6j-1)}{j! (216)^{j}}   ,
\eeq
and
\beq\label{airy-coef3}
d_j = c_j + c_{j-1} ( j - 5/6) = - \frac{6j+1}{6j-1} c_j ,
\eeq
with $c_0 = 1$ so that $d_0=1$. We note that the coefficients $(c_j, d_j)$
are related to the coefficients $(u_j, v_j)$ used by Olver as $c_j = (-1)^j u_j$ and
$d_j = (-1)^j v_j$.


\subsection{Uniform Asymptotics near the Eye-Shaped Region $K$}\label{unifasymp2}

Of particular importance is the application of the uniform asymptotics that follow
from (\ref{besselasympt1})-(\ref{airy1}) for $z$ near the eye-shaped region $K$ defined
in section \ref{sec:summary-zero1}. Recall that for a fixed $\epsilon > 0$,
we define a neighborhood $\Omega_{1, \epsilon}$ of $K$
to be all $z \in \C^+$ so that
$\text{dist}\; (z,\partial K^+)< \epsilon$ excluding small neighborhoods of $\pm 1$ so
that $z \in \Omega_{1, \epsilon}$ satisfies $|z+ 1|>\epsilon$ and $|z-1| > \epsilon$.

For $z \in \Omega_{1,\epsilon}$, the uniform asymptotics of the Bessel functions and
their derivatives follow from (\ref{besselasympt1}) and the estimates
on the Airy function and its derivative away from their zeros and given in (\ref{airy1}).
The uniform expansion of $J_\nu (\nu z)$
and its derivative involve the series $F_1 (\zeta)$, $F_2 (\zeta)$,
given in (\ref{coef1}), and $\hat{F}_1 (\xi)$ and $\hat{G}_1 (\xi)$,
defined in (\ref{airy-coef1}).
For the Bessel functions, we use $\xi = \nu \rho$ in the
expansion (\ref{airy1}).
We obtain:
\bea\label{besselasymptot2}
J_\nu ( \nu {z} ) &=& \frac{\phi (\zeta) e^{- \nu \rho}}{ 2 \pi^{1/2} \nu^{1/2} \zeta^{1/4}}
\left\{  1 + \frac{1}{\nu} \left(  \frac{c_1}{{\rho}} - \zeta^{1/2} B_0(\zeta) \right) \right. \nonumber \\
 &&\left. + \frac{1}{\nu^2} \left( A_1 (\zeta) + \frac{c_2}{\rho^2} - \frac{ \zeta^{1/2} B_0 (\zeta) d_1}{\rho} \right)
  + \mathcal{O} \left( \frac{1}{\nu^3} \right) \right\} ,
\eea
and for the derivative,
\bea\label{besselasymptot3}
J_\nu ' ( \nu {z} ) &=& - \frac{\psi (\zeta) e^{- \nu \rho}}{ 2 \pi^{1/2} \nu^{1/2} \zeta^{1/4}}
\left\{  - \zeta^{1/2} + \frac{1}{\nu} \left(  C_0 (\zeta) - \frac{ \zeta^{1/2} d_1}{{\rho}}  \right) \right. \nonumber
 \\
  &&\left. +  \frac{1}{\nu^2} \left( \frac{C_0 (\zeta) c_1}{\rho} - \frac{\zeta^{1/2} d_2}{\rho^2} - \zeta^{1/2} D_1 (\zeta) \right)
   +  \mathcal{O} \left( \frac{1}{\nu^3} \right) \right\}  ,
\eea
where the coefficients $C_j (\xi)$ and $D_j (\xi)$ are
defined in (\ref{hankel-coef1}) and (\ref{hankel-coef2}).

We also need the asymptotics for the Hankel functions of the first kind and their derivatives
for $z \in \Omega_{1,\epsilon}$. These are obtained from (\ref{hankelasympt1}) using the
expansions of the Airy functions. We note that for $z \in \Omega_{1, \epsilon}$,
it follows from the map $z \rightarrow \zeta$ that $\arg \zeta \sim - \pi / 3$.
Consequently, for $H_\nu^{(1)} ( \nu z)$, we have $\arg ( e^{2 \pi i / 3} \zeta ) \sim  \pi/ 3$,
and we need the asymptotics of the Airy function away from its zeros given in (\ref{airy1}).
Using these asymptotics (\ref{airy1}),
we find
for the Hankel functions of the first kind, with $\xi=-\nu \rho$
\bea\label{hankelasymptot2}
H_\nu^{(1)}(\nu z)  &=&    \frac{-i \phi ( \zeta) e^{\nu \rho}}{ \pi^{1/2}  \zeta^{1/4} \nu^{1/2} }
\left\{ 1 + \frac{1}{\nu}  \left( \zeta^{1/2} B_0 (\zeta) - \frac{c_1}{\rho} \right) \right. \nonumber \\
&& \left. + \frac{1}{\nu^2} \left( A_1(\zeta) + \zeta^{1/2} B_1 ( \zeta) +
 \frac{c_2}{\rho^2} - \frac{\zeta^{1/2} d_1 B_0(\zeta) }{\rho} \right)
 +  \mathcal{O} \left( \frac{1}{\nu^3}  \right)  \right\}.
\eea
As for the derivative, we use the same expansions of the Airy functions in (\ref{hankel-deriv-asympt1}), and obtain
\bea\label{hankelasymptot3}
{{H_\nu^{(1)}}} ' (\nu z)  &= &  \frac{i \psi ( \zeta) e^{ \nu \rho}}{ \pi^{1/2}  \zeta^{1/4} \nu^{1/2}}
\left\{ \zeta^{1/2} + \frac{1}{\nu}
\left( C_0 (\zeta) - \frac{d_1 \zeta^{1/2}}{\rho} \right) \right. \nonumber \\
 && \left. + \frac{1}{\nu^2} \left( \frac{-C_0(\zeta) c_1}{\rho} + \frac{\zeta^{1/2} d_2}{\rho^2}
 + \zeta^{1/2} D_1 (\zeta) \right) +  \mathcal{O} \left( \frac{1}{\nu^3}   \right) \right\}.
\eea


\subsection{Auxiliary Asymptotic Expansions}\label{subsec:aux1}

The variable $\tilde{z}(z) = [ z^2 - V_0 / \nu^2 ]^{1/2}$ carries the information
about the perturbing potential with strength $V_0$. We need the expansions as $\nu \rightarrow \infty$ of
various quantities depending on $\tilde{z}$ for $z$ in a fixed set $\Omega_{1, \epsilon}$
for which $z$ is bounded away from $\pm 1$ and $0$. First, we note that
\beq\label{eq:auxexp1}
\tilde{z} (z) = z - \frac{ V_0}{2 \nu^2 z} + \mathcal{O} \left( \frac{1}{\nu^4} \right) .
\eeq
Recalling the definition of $\rho$ in (\ref{map1}), we find
\beq\label{eq:auxexp4}
\rho ( \tilde{z}) = \rho (z) + \frac{V_0}{2 \nu^2} \left[ \frac{(1-z^2)^{1/2}}{z^2} \right] +
\mathcal{O} \left( \frac{1}{\nu^4} \right) .
\eeq
It follows from the definition of $\zeta$ in (\ref{map1}) that
\bea\label{eq:auxexp5}
\frac{\zeta (\tilde{z}) }{\zeta (z)} &=& 1 + \frac{V_0}{3 \nu^2} \left[ \frac{ \rho'}{ z \rho} \right]
+  \mathcal{O} \left( \frac{1}{\nu^4} \right) \nonumber \\
 &=&1 + \frac{V_0}{2 \nu^2 \zeta^{3/2}} \left[ \frac{ (1 - z^2)^{1/2} }{ z^2 } \right]
+  \mathcal{O} \left( \frac{1}{\nu^4} \right) \ .
\eea
We write $\tilde{\rho}$ and $\tilde{\zeta}$ for $\rho
( \tilde{z})$ and $\zeta (\tilde{z})$.
We need the expansion of the following combination
that follows from (\ref{eq:auxexp4})--(\ref{eq:auxexp5}):
\beq\label{eq:auxexp7}
\frac{\phi (\tilde{\zeta}) }{\phi (\zeta)} \left( \frac{\zeta}{\tilde{\zeta}} \right)^{1/4}
= 1 - \frac{V_0}{ 4 \nu^2 (1-z^2)} + \mathcal{O} \left( \frac{1}{\nu^4} \right) .
\eeq
Finally, we need an asymptotic expansion for the exponentials
 appearing in the products of Bessel
functions, see (\ref{besselasymptot2})--(\ref{besselasymptot3}).
This follows from
(\ref{eq:auxexp4}):
\beq\label{eq:auxexp8}
e^{- \nu ( \tilde{\rho} - \rho)} = 1 - \frac{ V_0 (1-z^2)^{1/2}}{2 \nu z^2}
+ \frac{V_0^2 (1-z^2)}{8\nu^2 z^4}
+ \mathcal{O} \left( \frac{1}{\nu^3}
\right) .
\eeq

\subsection{Condition for Zeros on the $m^{th}$-Sheet}\label{subsec:condition-m-0}

We now compute the asymptotic expansion of $F_0^{(\nu)}(\nu z)$, defined in (\ref{zero-sheet1}),
to order $\nu^{-2}$, and of $G_0^{(\nu)} (\nu z)$, defined in (\ref{m-sheet21}),
to order $\nu^{-3}$, using the uniform asymptotics of section \ref{unifasymp2}.
As for $F_0^{(\nu)}(\nu z)$, recall that
\beq\label{physheet1a}
F_0^{(\nu)}(\nu z) = \nu \tilde{z}(z) J_\nu ' (\nu \tilde{z}(z)) H_\nu^{(1)}(\nu z) - \nu z
J_\nu ( \nu \tilde{z}(z)) {H_\nu^{(1)}}'(\nu z) .
\eeq
From the asymptotics of the Bessel functions (\ref{besselasymptot2})--(\ref{besselasymptot3})
and of the Hankel functions of the first kind (\ref{hankelasymptot2})--(\ref{hankelasymptot3}),
together with auxiliary expansions (\ref{eq:auxexp4})--(\ref{eq:auxexp5}),
we find that the first term on the right in (\ref{physheet1a}) has the
expansion in inverse powers of $\nu$:
\bea\label{m-sheet331}
\lefteqn{ \nu \tilde{z} (z) J_\nu ' (\nu \tilde{z} (z) ) H_{\nu}^{(1)} (\nu z ) } && \nonumber \\
 &= & \frac{ i}{ \pi} \left\{ -1 + \frac{1}{\nu} \left[ \frac{V_0 (1-z^2)^{1/2}}{2 z^2} +
 \frac{C_0 (\zeta)}{\zeta^{1/2}} + \frac{c_1 - d_1 }{ \rho} - \zeta^{1/2} B_0 (\zeta)  \right]
 + \mathcal{O} \left(\frac{1}{\nu^2} \right) \right\} . \nonumber \\
\eea
Similarly, we find for the second term on the right in (\ref{physheet1a}),
we the expansion in $\nu$ is
\bea\label{m-sheet221}
\lefteqn{ \nu z J_\nu (\nu \tilde{z} (z) )  {H_{\nu}^{(1)}}' (\nu z )  } && \nonumber \\
&=&
\frac{ i}{ \pi} \left\{ 1 + \frac{1}{\nu} \left[ \frac{-V_0 (1-z^2)^{1/2}}{2 z^2} +
 \frac{C_0 (\zeta)}{\zeta^{1/2}} - \zeta^{1/2} B_0 (\zeta) + \frac{ c_1 - d_1}{ \rho}  \right]
  \nonumber
+ \mathcal{O} \left(\frac{1}{\nu^2} \right) \right\} .  \\
\eea

Subtracting the expansion (\ref{m-sheet221}) from (\ref{m-sheet331}), we find that $F_0^{(\nu)}(\nu z)$
has the asymptotic form
\bea\label{mzeros01a}
F_0^{(\nu)} (\nu z) &=& \frac{- 2i}{\pi} \left\{ 1 - \frac{1}{\nu}
 \left[ \frac{V_0 (1-z^2)^{1/2}}{2 z^2}
\right]
+ \mathcal{O} \left( \frac{1}{\nu^2} \right) \right\} .
\eea
This is (\ref{mzeros01b}). Note that this implies that, asymptotically, there are no zeros on the
physical sheet, as expected.

We turn to the second quantity $G_0^{(\nu)}(\nu z)$ defined in (\ref{m-sheet21}):
\beq\label{eq:defng1}
G_0^{(\nu)} (\nu z) = \nu \tilde{z} (z) J_\nu ' ( \nu \tilde{z}) J_\nu (\nu z)
- \nu z J_\nu  ( \nu \tilde{z}) J_\nu '(\nu z) .
\eeq
We treat each term separately. For the first term,
we find
the following asymptotic expansion in $\nu$:
\bea\label{m-sheet441}
\lefteqn{\nu \tilde{z} (z) J_\nu ' ( \nu \tilde{z}) J_\nu (\nu z) } && \nonumber \\
 &=& \left( \frac{-1}{2 \pi} \right) e^{- \nu (\rho+\tilde{\rho})} \left\{ -1 + \frac{1}{\nu} \left[ \frac{C_0(\zeta)}{\zeta^{1/2}} -
  \frac{c_1 + d_1}{\rho} + \zeta^{1/2} B_0(\zeta)
 \right] \right. \nonumber \\
  && \left. +\frac{1}{\nu^2}
\left[ \frac{2 c_1 C_0 (\zeta)}{\rho \zeta^{1/2}} -
  \frac{d_1 c_1 + d_2 + c_2}{ \rho^2} -
  ( B_0(\zeta)C_0(\zeta) +  A_1 (\zeta) + D_1(\zeta)) + \frac{2 \zeta^{1/2} B_0(\zeta) d_1}{\rho}
   \right. \right. \nonumber \\
   && \left. \left.
- \frac{V_0}{4 (1-z^2)} \right]  +
   \mathcal{O} \left(\frac{1}{\nu^3} \right) \right\} . \nonumber \\
\eea
For the second term in $G_0^{(\nu)}(\nu z)$, we obtain
\bea\label{m-sheet551}
\lefteqn{\nu z J_\nu  ( \nu \tilde{z}) J_\nu '(\nu z) } && \nonumber \\
 &=& \left( \frac{-1}{2 \pi} \right) e^{- \nu( \rho+\tilde{\rho})}
 \left\{ -1 + \frac{1}{\nu} \left[ \frac{C_0(\zeta)}{\zeta^{1/2}} -
  \frac{c_1 + d_1}{\rho} + \zeta^{1/2} B_0(\zeta)
\right] \right. \nonumber \\
  && \left. +\frac{1}{\nu^2} \left[ \frac{2 c_1 C_0 (\zeta)}{\rho \zeta^{1/2}} -
  \frac{d_1 c_1 + d_2 + c_2}{ \rho^2} -
  (B_0(\zeta)C_0(\zeta) + A_1 (\zeta) + D_1(\zeta)) + \frac{2 \zeta^{1/2} B_0(\zeta) d_1}{\rho}
   \right. \right. \nonumber \\
   && \left. \left.
+ \frac{V_0}{4 (1-z^2)} \right]  +
   \mathcal{O} \left(\frac{1}{\nu^3} \right) \right\} . \nonumber \\
\eea

Subtracting (\ref{m-sheet551}) from (\ref{m-sheet441}), we obtain the expansion for $G_0^{(\nu)}(\nu z)$:
\beq\label{mzeros01c}
G_0^{(\nu)}(\nu z) =\frac{e^{-  \nu (\rho+\tilde{\rho})}}
{2 \pi} \left[ \frac{V_0}{ 2 \nu^2 (1-z^2)} + \mathcal{O}
\left( \frac{1}{\nu^3} \right) \right] = \frac{e^{- 2 \nu \rho}}{2 \pi} \left[ \frac{V_0}{ 2 \nu^2 (1-z^2)} + \mathcal{O}
\left( \frac{1}{\nu^3} \right) \right] .
\eeq
Note that when $V_0 = 0$, this term vanishes to higher order in $\nu$.
This establishes (\ref{m-sheet3}).


\small
\noindent
{\sc
Department of Mathematics\\
University of Missouri\\
Columbia, Missouri 65211\\
e-mail:{\tt tjc@math.missouri.edu} }

\vspace{2mm}

\noindent
{\sc Department of Mathematics\\
University of Kentucky\\
Lexington, Kentucky 40506-0027\\
e-mail: {\tt hislop@ms.uky.edu}}

\end{document}